\documentclass[english,onecolumn]{article}
\usepackage[utf8]{inputenc}%(only for the pdftex engine)
\usepackage[big]{dgruyter}

\usepackage{comment,xcomment}
\usepackage{booktabs,centernot}
\usepackage{amssymb}
\usepackage{graphicx,float,pdflscape,multicol,tabularx}
\usepackage{soul}
\usepackage[labelfont={small,bf},font=small,width=.92\textwidth]{caption}
\usepackage{color,xcolor,tcolorbox}
\definecolor{lightblue}{RGB}{173,216,230}
\usepackage{xfrac,ifthen}
\usepackage{enumitem}
\usepackage{adjustbox}

\newcolumntype{L}{>{\hsize=1\hsize\raggedright\arraybackslash}X}
\newcolumntype{Y}{>{\hsize=.4\hsize\raggedright\arraybackslash}X}
\newcolumntype{Z}{>{\hsize=1.6\hsize\raggedright\arraybackslash}X}
\newcommand{\undentLength}{-10pt}
\newcommand{\undent}{\hspace*{\undentLength}}
\newcommand{\Exp}{\mathbb{E}}

\newcommand{\expit}{\mathrm{expit}}
\newcommand{\logit}{\mathrm{logit}}

\newcommand{\CI}{\mathrel{\perp\mspace{-10mu}\perp}}
\newcommand{\nCI}{\centernot{\CI}}

\newcounter{append}
\renewcommand{\theappend}{\Alph{append}}
\newenvironment{myAppendix}[1][]{\refstepcounter{append}\section*{\Large{\textsc{Appendix~\theappend#1}}}}{}
\newcounter{assump}

\newenvironment{myAssumption*}[1][]{\begin{quotation}\textit{Assumption\ifthenelse{\equal{#1}{}}{.}{~#1.}\quad}}{\end{quotation}}
\newcounter{defin}

\newenvironment{myDefinition*}[1][]{\begin{quotation}\textit{Definition\ifthenelse{\equal{#1}{}}{.}{~#1.}\quad}}{\end{quotation}}
\newtheorem{theorem}{Theorem}%[section]
\newtheorem{corollary}{Corollary}%[section]
\newtheoremstyle{named}{}{}{\itshape}{}{\bfseries}{.}{.5em}{\thmnote{#3}}
\theoremstyle{named}

\newtheorem*{remark}{Remark}

\renewcommand\bibname{References}

\newcommand{\TableMatching}{Table 2}
\newcommand{\TableNoMatching}{Table 1}
\newcommand{\BaseSamplingNoMatching}{S1}
\newcommand{\SurvivorSamplingNoMatching}{S2}
\newcommand{\RiskSetSamplingNoMatching}{S3}
\newcommand{\PPRiskSetSamplingNoMatching}{S4}
\newcommand{\HomogeneousRiskRatios}{H1}
\newcommand{\HomogeneousRates}{H2}
\newcommand{\HomogeneousConditionalRates}{H3}
\newcommand{\PPHomogeneousRates}{H4}
\newcommand{\BaseSamplingMatching}{M1}
\newcommand{\SurvivorSamplingMatching}{M2}
\newcommand{\RiskSetSamplingMatching}{M3}
\newcommand{\PPRiskSetSamplingMatching}{M4}
\newcommand{\HomogeneousOddsRatios}{H5}
\newcommand{\HomogeneousRateRatios}{H6}
\newcommand{\PPHomogeneousRateRatios}{H7}

\newcommand{\mytitle}{Identification of causal effects in case-control studies}

\newcommand{\Item}{{\small $\bullet$}~}

\begin{document}

%\DoubleSpacing
\allowdisplaybreaks

% \thispagestyle{empty}
% %{\color{red}\noindent To the editorial board of the \journal,\bigskip}
% \noindent To the editors of \journal,\bigskip\medskip

% \noindent Enclosed, please find our manuscript entitled `\mytitle', which we would like to submit as an Original Research Article to \journal.\bigskip

% \noindent In this methodological paper, we give an overview of how causal effects can be identified in case-control studies, and in doing so, recast historical concepts, assumptions and principles in a modern and formal framework. We cover various estimands relating to time-varying treatments, popular sampling schemes (case-base, survivor, and risk-set sampling) as well as designs with and designs without matching. We believe that our approach and articulation of estimands, assumptions and identification strategies might resolve or prevent misunderstanding and offer a useful starting point for future research on the topic of case-control designs. \bigskip

% \noindent Since case-control designs are a popular yet commonly misunderstood tool, we believe the paper is of interest to the broad readership of \journal. Any consideration of this work for publication will be much appreciated.\bigskip\medskip

% \noindent Yours sincerely,\medskip

% \noindent Bas Penning de Vries and Rolf Groenwold\bigskip

% \noindent{\bfseries{Reference:}\\}
% \noindent VanderWeele, T.~J., A.~R. Luedtke, M.~J. van~der Laan, and R.~C. Kessler
%   (2019): ``Selecting optimal subgroups for treatment using many
%   covariates,'' \emph{Epidemiology}, 30, 334--341.

\newpage

\articletype{Research Article{\hfill}}%Open Access}

\author*[1]{Bas B.L. Penning de Vries}

\author[2]{Rolf H.H. Groenwold}

\affil[1]{Department of Clinical Epidemiology, Leiden University Medical Center, PO Box 9600, 2300 RC, Leiden, The Netherlands; telephone: +31~71 526~5639; e-mail: B.B.L.Penning\_de\_Vries@lumc.nl}

\affil[2]{Departments of Clinical Epidemiology and Biomedical Data Sciences, Leiden University Medical Center, Leiden, The Netherlands}

\title{\huge \mytitle}

\runningtitle{Identification in case-control studies}

%\subtitle{...}

\begin{abstract}
{Case-control designs are an important tool in contrasting the effects of well-defined treatments. In this paper, we reconsider classical concepts, assumptions and principles and explore when the results of case-control studies can be endowed a causal interpretation. Our focus is on identification of target causal quantities, or estimands. We cover various estimands relating to intention-to-treat or per-protocol effects for popular sampling schemes (case-base, survivor, and risk-set sampling), each with and without matching. Our approach may inform future research on different estimands, other variations of the case-control design or settings with additional complexities.}
\end{abstract}
%\classification[MSC]{Please put MSC 2010 codes here.}
 % \communicated{...}
 % \dedication{...}

  \journalname{Journal of Causal Inference}
% \DOI{DOI}
%   \startpage{1}
%   \received{..}
%   \revised{..}
%   \accepted{..}

  \journalyear{2019}
  \journalvolume{1}
%  \journalissue{1}

\maketitle

\thispagestyle{empty}

% \section*{Target journal}
% Journal of Causal Inference

% \section*{Type of manuscript}
% Original Research Article (methodological).

\newpage
%\setcounter{page}{1}
% \section*{Abstract}
% Case-control designs are an important tool in contrasting the effects of well-defined treatments. In this paper, we reconsider classical concepts, assumptions and principles and explore when the results of case-control studies can be endowed a causal interpretation. Our focus is on identification of target causal quantities, or estimands. We cover various estimands relating to intention-to-treat or per-protocol effects for popular sampling schemes (case-base, survivor, and risk-set sampling), each with and without matching. Our approach may inform future research on different estimands, other variations of the case-control design or settings with additional complexities.%We argue that a formal approach to case-control studies that is built on the modern arsenal for identifying causal effects has the potential to improve the quality of studies and resolve or prevent misunderstanding.

% \section*{Key words} Causal inference, case-control designs, identifiability

%\newpage

% {\color{red}
% \section*{Reading list}
% %\citeauthor{He2020} (\citeyear{He2020})\newline
% % \citeauthor{Labrecque2021} (\citeyear{Labrecque2021})\newline
% %\noindent \citeauthor{Hernan2015} (\citeyear{Hernan2015}) on %person-time\newline
% %\noindent \citeauthor{Knol2008} (\citeyear{Knol2008})\newline
% %\noindent \citeauthor{Guess2006} (\citeyear{Guess2006}) \newline
% %\noindent \citeauthor{Lefebvre2006} (\citeyear{Lefebvre2006}) \newline
% %\noindent \citeauthor{Robins1999} (\citeyear{Robins1999})
% }

\section{Introduction}

In causal inference, it is important that the causal question of interest is unambiguously articulated \citep{Hernan2020book}. The causal question should dictate, and therefore be at the start of, investigation. When the target causal quantity, the estimand, is made explicit, one can start to question %whether---and if so, how---
%it can be estimated from the available data.
how
it relates to the available data distribution and, as such, form a basis for estimation with finite samples from this distribution.

The counterfactual framework offers a language rich enough to articulate a wide variety of causal claims that can be expressed as what-if statements \citep{Hernan2020book}. Another, albeit closely related, approach to causal inference is target trial emulation, an explicit effort to mitigate departures from a study (the `target trial') that, if carried out, would enable one to readily answer the causal what-if question of interest \citep{Hernan2016}. While it may be too impractical or unethical to implement, making explicit what a target trial looks like has particular value in communicating the inferential goal and offers a reference against which to compare studies that have been or are to be conducted.

The counterfactual framework and emulation approach have become increasingly popular in %the conduct of
observational cohort studies. Case-control studies, however, have not yet enjoyed this trend. A notable exception is given by \citeauthor{Dickerman2020} (\citeyear{Dickerman2020}), who recently outlined an application of trial emulation with case-control designs to statin use and colorectal cancer.

% A target trial is a study that, if carried out, would enable one to readily answer the causal question of interest \citep{Hernan2016}. While it may be too impractical or unethical to implement, making explicit a target trial has particular value in communicating the inferential goal and offers a reference against which to compare studies that have been or are to be implemented. When a target trial fails to materialise, evidence may be drawn from a different trial or from observational cohort or case-control studies. Regardless of the study's design, it is important to account for discrepancies between the target trial and the study from which evidence is actually drawn.
% }

%Target trial emulation is an explicit effort to mitigate departures from the target trial \citep{Hernan2016}. 

In this paper, we give an overview of how observational data obtained with case-control designs can be used to identify
%several
a number of 
causal estimands  %, the target causal quantities, whose values would be evident if the respective target trials were implemented. 
%That is, for a number of case-control designs and estimands, we describe how and under which conditions the observed or available data distribution is compatible with exactly one value of the estimand. 
and, in doing so, recast historical case-control concepts, assumptions and principles in a modern and formal %the counterfactual outcomes 
framework.% and further %illustrate
%discuss
%how the identification results naturally translate into estimators. %relate historical case-control concepts and principles to the counterfactual outcomes and trial emulation framework.

\section{Preliminaries}%\section{Notation and set-up}
\subsection{Identification versus estimation}
%In what follows, we will discuss a number of estimands. 
An estimand is said to be identifiable if the distribution of the available data is compatible with exactly one value of the estimand, or therefore, if the estimand can be expressed as a function of the available data distribution. Identification forms a basis for estimation with finite samples from this distribution \citep{Petersen2014}. Once the estimand has been made explicit and an identifiability expression established, estimation is a purely statistical problem. %As we will find, 
While the expression will often naturally translate into a plug-in estimator, there is, however, generally more than one way to translate an identifiability result into an estimator and different estimators may have important differences in their statistical properties. %For appealing finite-sample properties, additional assumptions (i.e., beyond the identification assumptions) are needed. For example, finite-sample properties will typically rest on no model misspecification and that variables are independently distributed among individuals and drawn from the same theoretical (or `population') distribution. By contrast, 
%The identification results that we present are mostly non-parametric and relate to the theoretical data distribution itself and not to samples from it. %Estimation also requires precision estimation, while for identification, sampling variability is irrelevant. 
Here, our focus is on identification, so that the purely statistical issues of the next step in causal inference, estimation, can be momentarily put aside.% \citep{Petersen2014}.

\subsection{Case-control study nested in cohort study}
To facilitate understanding, it is useful to consider every case-control study as being ``nested'' within a cohort study. A case-control study is effectively a cohort study with missingness governed by the control sampling scheme. Therefore, when the observed data distribution of a case-control study is compatible with exactly one value of a given estimand, then so is the available or observed data distribution of the underlying cohort study. In other words, identifiability of an estimand with a case-control study implies identifiability of the estimand with the cohort study within which it is nested. % within which the case-control study is nested.``... a failure to anchor this to an underlying cohort study that explicitly emulates a target trial.'' (Dickerman et al. 2020)
The converse is not evident and in fact may not be true. In this paper, the focus is on sets of conditions or assumptions that are sufficient for identifiability in case-control studies.

\begin{figure}[t]\centering
\caption{Illustration of possible courses of follow-up of an individual for a study with baseline $t_0$ and administrative study end $t_{12}$. }\label{fig:setupIllustration}
\includegraphics{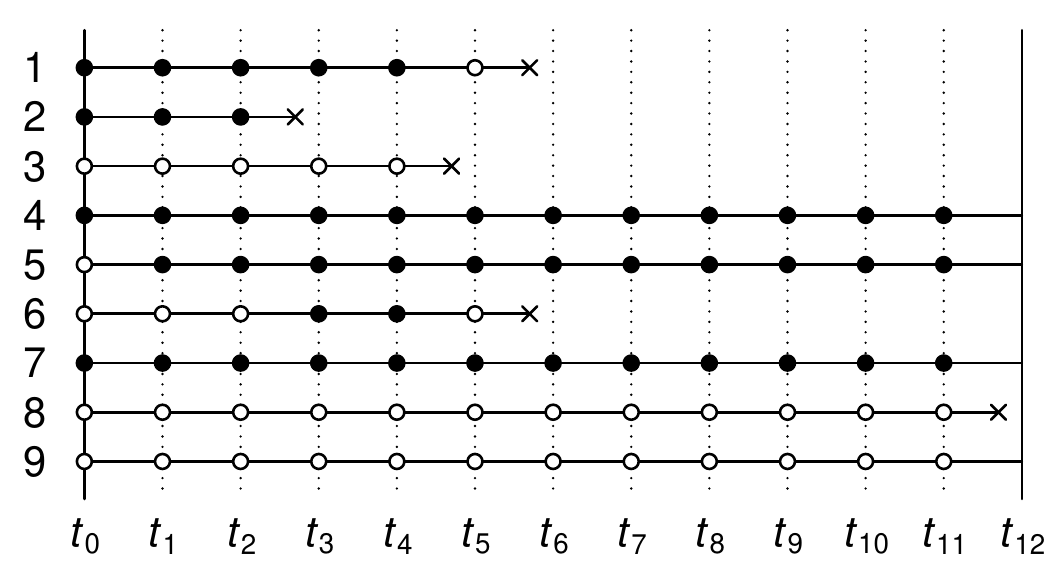}
\caption*{Solid bullets indicate `exposed'; empty bullets indicate `not exposed'. The incident event of interest is represented by a cross.}
\end{figure}

\subsection{Set-up of underlying cohort study}
Consider a time-varying exposure $A_k$ that can take one of two levels, 0 or 1, at $K$ successive time points $t_k$ ($k=0,1,...,K-1$), where $t_0$ denotes baseline (cohort entry or time zero). Study participants are followed over time until they sustain the event of interest or the administrative study end $t_K$, whichever comes first. We denote by $T$ the time elapsed from baseline until the event of interest and let $Y_k=I(T<t_k)$ indicate whether the event has occurred by $t_k$. The lengths between the time points are typically fixed at a constant (e.g., of one day, week, or month). Figure~\ref{fig:setupIllustration} depicts twelve equally spaced time points over, say, twelve months with several possible courses of follow-up of an individual. As the figure illustrates, individuals can switch between exposure levels during follow-up, as in any truly observational study. Apart from exposure and outcome data, we also consider a (vector of) covariate(s) $L_k$, which describes time-fixed individual characteristics or time-varying characteristics typically relating to a time window just before exposure or non-exposure at $t_k$, $k=0,1,...,K-1$.

\subsection{Causal contrasts}
%As mentioned, it is important that the causal contrast of interest is made explicit. 
Although there are many possible contrasts, particularly with time-varying exposures, for simplicity we consider only two pairs of mutually exclusive interventions: (1) setting baseline exposure $A_0$ to 1 versus 0; and (2) setting all of $A_0,A_1,...,A_{K-1}$ to 1 (`always exposed') versus all to 0 (`never exposed'). For $a=0,1$, we let counterfactual outcome $Y_k(a)$ indicate whether the event has occurred by $t_k$ under the baseline-only intervention that sets $A_0$ to $a$. By convention, we write $\overline{1}=(1,1,...,1)$ and $\overline{0}=(0,0,...,0)$, and let $Y_k(\overline{1})$ and $Y_k(\overline{0})$ indicate whether the event has occurred by $t_k$ under the intervention that sets $(A_0,A_1,...,A_{K-1})$ all to 1 and all to 0, respectively.
%These outcomes of possibly contrary-to-fact settings are called counterfactual outcomes. 
Further details about the notation and set-up are given in Supplementary Appendix~\ref{app:notation}.

\subsection{Case-control sampling}
The fact that each time-specific exposure variable can take only one value per time point means that at most one counterfactual outcome can be observed per individual. This type of missingness is common to all studies. Relative to the cohort studies within which they are nested, case-control studies have additional missingness, which is governed by the control sampling scheme. In this paper, we focus on three well-known sampling schemes: case-base sampling, survivor sampling, and risk-set sampling. The next sections gives an overview of conditions under which intention-to-treat and always-versus-never-exposed per-protocol effects can be identified with the data that are observed under these sampling schemes.

\section{Case-control studies without matching}
\TableNoMatching{} summarises a number of identification results for case-control studies without matching. More formal statements and proofs are given in Supplementary Appendix~\ref{app:noMatching}. In all case-control studies that we consider in this section, cases are compared with controls with regard to their exposure status via an odds ratio, even when an effect measure other than the odds ratio is targeted. 
An individual qualifies as a case if and only if they sustain the event of interest by the administrative study end (i.e., $Y_K=1$) and adhered to one of the protocols of interest until the time of the incident event. In Figure~\ref{fig:setupIllustration}, %incident events are marked with a cross. The 
the individual represented by row 1 is therefore regarded as a case (an exposed case in particular) in our investigation of intention-to-treat effects but not in that of per-protocol effects.
Whether an individual (also) serves as a control depends on the control sampling scheme.

\subsection{Case-base sampling}
The first result in \TableNoMatching{} describes how to identify the intention-to-treat effect as quantified by the marginal risk ratio 
\begin{align*}
\frac{\Pr(Y_K(1)=1)}{\Pr(Y_K(0)=1)}
\end{align*}
under case-base sampling.
(For identification of a conditional risk ratio, see Theorem~\ref{th:BaseSamplingConditionalITT} of Supplementary Appendix~\ref{app:noMatching}.)
Case-base sampling, also known as case-cohort sampling, means that no individual who is at risk at baseline of sustaining the event of interest is precluded from selection as a control. Selection as a control, $S$, is further assumed independent of baseline covariate $L_0$ and exposure $A_0$.  %and has probability equal to a constant known as the sampling fraction. %, $\delta$:
% \begin{align*}
% \Pr(S=1|L_0,A_0)=\Pr(S=1)=\delta. \tag{\BaseSamplingNoMatching{}}
% \end{align*}
Selecting controls from survivors only (e.g., rows 4, 5, 7 and 9 in Figure~\ref{fig:setupIllustration}) violates this assumption when survival depends on $L_0$ or $A_0$. 

To account for baseline confounding, inverse probability weights could be derived from control data according to
\begin{align} \label{eq:weights1}
W &= \frac{A_0}{\Pr(A_0=1|L_0,S=1)}+\frac{1-A_0}{1-\Pr(A_0=1|L_0,S=1)}.
\end{align}
We then compute the odds of baseline exposure among cases and among controls in the pseudopopulation that is obtained by weighting everyone by subject-specific values of $W$. The ratio of these odds coincides with the target risk ratio under the three key identifiability conditions of consistency, baseline conditional exchangeability and positivity \citep{Hernan2020book}.

The identification result for case-base sampling suggests a plug-in estimator: %simply
replace all functionals of the theoretical data distribution with sample analogues. For example, to obtain the weight for an individual with baseline covariate level $l_0$, replace the theoretical propensity score $\Pr(A_0=1|L_0=l_0,S=1)$ with an estimate $\widehat{\Pr}(A_0=1|L_0=l_0,S=1)$ derived from a fitted model (e.g., a logistic regression model) that imposes parametric constraints on the distribution of $A_0$ given $L_0$ among the controls.

%\newpage
\pagebreak
%\SingleSpacing
\global\pdfpageattr\expandafter{\the\pdfpageattr/Rotate 90} % to rotate page in viewer

\rotatebox{90}{\hspace{3pt}\begin{minipage}{\textheight}
\textbf{\TableNoMatching}. Overview of (non-parametric) identification results for case-control studies without matching.\end{minipage}}
\rotatebox{90}{
\begin{tabularx}{\textheight}{YLLL}
Sampling scheme & Estimand & Assumptions & Identification strategy \\ \midrule
% %%%%%%%%%%%%%%%%%%%
% % Case-base
% %%%%%%%%%%%%%%%%%%%
Case-base & Risk ratio for intention-to-treat effect

\medskip\centering{$\displaystyle\frac{\Pr(Y_K(1)=1)}{\Pr(Y_K(0)=1)}$} & 
\Item 
Control selection $S$ independent of baseline covariates $L_0$ and exposure $A_0$
\newline\Item Consistency \newline\Item Baseline exchangeability given $L_0$\newline\Item Positivity \newline (Theorem~\ref{th:BaseSamplingMarginalITT})
& 1. Derive time-fixed IP weights $W$ from control data\newline 2. Compute the baseline exposure odds among cases, weighted by $W$\newline 3. Compute the baseline exposure odds among controls, weighted by $W$
\\\cmidrule{1-4}
% %%%%%%%%%%%%%%%%%%%
% % Survivor
% %%%%%%%%%%%%%%%%%%%
Survivor & Odds ratio for intention-to-treat effect 

\medskip\centering{$\displaystyle\frac{\mathrm{Odds}(Y_K(1)=1|L_0)}{\mathrm{Odds}(Y_K(0)=1|L_0)}$}
& \Item Control selection $S$ independent of baseline exposure $A_0$ given baseline covariates $L_0$ and survival until $t_K$ ($Y_K=0$) 
\newline\Item Consistency \newline\Item Baseline exchangeability given $L_0$ \newline\Item Positivity \newline (Theorem~\ref{th:SurvivorSamplingConditionalITT})
& 1. Derive the conditional baseline exposure odds given $L_0$ among cases \newline 2. Derive the conditional baseline exposure odds given $L_0$ among controls\newline\ 3. Take the ratio of the results of steps 1 and 2
\\\cmidrule{1-4}
% %%%%%%%%%%%%%%%%%%%
% % Risk-set ITT
% %%%%%%%%%%%%%%%%%%%
Risk-set & Hazard ratio for intention-to-treat effect 

\medskip\centering{$\displaystyle\frac{\Pr(Y_{k+1}(1)=1|Y_k(1)=0)}{\Pr(Y_{k+1}(0)=1|Y_k(0)=0)}$}& \Item Control selection $S_k$ independent of baseline covariates $L_0$ and exposure $A_0$ given eligibility at $t_k$ ($Y_k=0$) with constant sampling probability among those eligible$^\dagger$
\newline\Item Consistency \newline\Item Baseline exchangeability given $L_0$ \newline\Item Positivity \newline\Item Constant counterfactual hazards \newline (Theorem~\ref{th:DensitySamplingMarginalITT})
&  1. Derive time-fixed IP weights $W$ from control data\newline 2. Compute baseline exposure odds among cases, weighted by $W$\newline 3. Compute baseline exposure odds among controls, weighted by $W$ times $\sum_{k=0}^{K-1}S_k$, the number of times selected as a control\newline 4. Take the ratio of the results of steps 2 and 3
\\\cmidrule{2-4}
% %%%%%%%%%%%%%%%%%%%
% % Risk-set PP
% %%%%%%%%%%%%%%%%%%%
& Hazard ratio for per-protocol effect 

\medskip\centering{$\displaystyle\frac{\Pr(Y_{k+1}(\overline{1})=1|Y_k(\overline{1})=0)}{\Pr(Y_{k+1}(\overline{0})=1|Y_k(\overline{0})=0)}$}
& {\Item Control selection $S_k$ independent of covariate and exposure history up to $t_k$ given eligibility at $t_k$ ($Y_k=0$) with constant sampling probability among those eligible$^\dagger$
\newline\Item Consistency\newline\Item Sequential conditional exchangeability\newline\Item Positivity \newline\Item Constant counterfactual hazards} \newline (Theorem~\ref{th:DensitySamplingPP})
& 1. Derive time-varying IP weights $W_k$ from control data\newline 2. Censor from time of protocol deviation \newline
3. Compute (baseline) exposure odds among cases, weighted by those weights $W_k$ such that $Y_k=0$ and $Y_{k+1}=1$ \newline4. Compute (baseline) exposure odds among all controls, weighted by $\sum_{k=0}^{K-1}W_kS_k$, the weighted number of times selected as a control \newline5. Take the ratio of the results of steps 3 and 4\\
\end{tabularx}
}
\rotatebox{90}{\hspace{3pt}\begin{minipage}{\textheight}
See text or Supplementary Material for elaboration on assumptions. $^{\dagger}$Weaker/alternative control selection assumptions are given in the Supplementary Material.\end{minipage}}

\pagebreak
\global\pdfpageattr\expandafter{\the\pdfpageattr/Rotate 0}

%\DoubleSpacing

\subsection{Survivor sampling}
With survivor (cumulative incidence or exclusive) sampling, a subject is eligible for selection as a control only if they reach the administrative study end event-free. To identify the conditional odds ratio of baseline exposure versus baseline non-exposure given $L_0$,
\begin{align*}
\frac{\mathrm{Odds}(Y_K(1)=1|L_0)}{\mathrm{Odds}(Y_K(0)=1|L_0)},
\end{align*}
selection as a control, $S$, is assumed independent of baseline exposure $A_0$ given $L_0$ and survival until the end of study (i.e., $Y_K=0$). 

The directed acyclic graph (DAG) of Figure~\ref{fig:DAG} is compatible with both survivor sampling and case-base sampling. For those well versed in DAGs, it is tempting to conclude from it that restricting the analysis to those included in the study, i.e., conditioning on study inclusion, would result in bias (or departure from identification), by way of collider stratification. Although conditioning on study inclusion may indeed induce an association between baseline exposure and unmeasured cause $U$ of $Y_K$ (within levels of $L_0$), it is important to recognise it need not result in bias \citep{Westreich2012,Hughes2019}.

In fact, as is shown in Supplementary Appendix~\ref{app:noMatching}, Theorem~\ref{th:SurvivorSamplingConditionalITT}, the above odds ratio is identified by the ratio of the baseline exposure odds given $L_0$ among the cases versus controls, provided the key identifiability conditions of consistency, baseline conditional exchangeability, and positivity are met.

% Marginal versus conditional odds ratio...:
All estimands in \TableNoMatching{} describe a marginal effect, except for the odds ratio, which is conditional on baseline covariates $L_0$. The corresponding marginal odds ratio
\begin{align*}
\frac{\mathrm{Odds}(Y_K(1)=1)}{\mathrm{Odds}(Y_K(0)=1)}
\end{align*}
is not identifiable from the available data distribution under the stated assumptions (see remark to Theorem~\ref{th:SurvivorSamplingConditionalITT}, Supplementary Appendix~\ref{app:noMatching}). However, approximate identifiability can be achieved by invoking the rare event assumption (or rare disease assumption), in which case the marginal odds ratio approximates the marginal risk ratio.

%Suppose that, given event-free survival until the admistrative study end ($Y_K=0$), selection as a control, $S$, is independent of exposure, covariate and outcome data, such that the probability of selection among survivors equals a constant, `the sampling fraction', $\delta$:
% \begin{align*}
% \Pr(S=1|Y_0,L_0,A_0,Y_1,...,L_{K-1},A_{K-1},Y_K)=\delta\times (1-Y_K). \tag{\SurvivorSamplingNoMatching{}}
% \end{align*}
% %According to \SurvivorSamplingNoMatching{}, 

% Then, the standard conditions of consistency for baseline interventions, conditional exchangeability given $L_0$ at baseline, and positivity together suffice for identification of the conditional odds ratio
% \begin{align*}
% \frac{\mathrm{Odds}(Y_K(1)=1|L_0)}{\mathrm{Odds}(Y_K(0)=1|L_0)}.
% \end{align*}
% Under these conditions, the odds ratio is identified simply by the ratio of exposure odds given baseline covariates $L_0$ among cases versus controls.

% {\color{red}[Comment on collider stratification]}

\begin{figure}[t]\centering
\caption{Directed acyclic graph for a setting where inclusion (as case or control) into the case-control study with case-base or survivor sampling is determined by the outcome variable $Y_K$. $U$ represents an unknown or unobserved cause of $Y_K$. The dashed double-headed arrow represents an unmeasured or observed common cause.
}\label{fig:DAG}
\includegraphics[scale=.7]{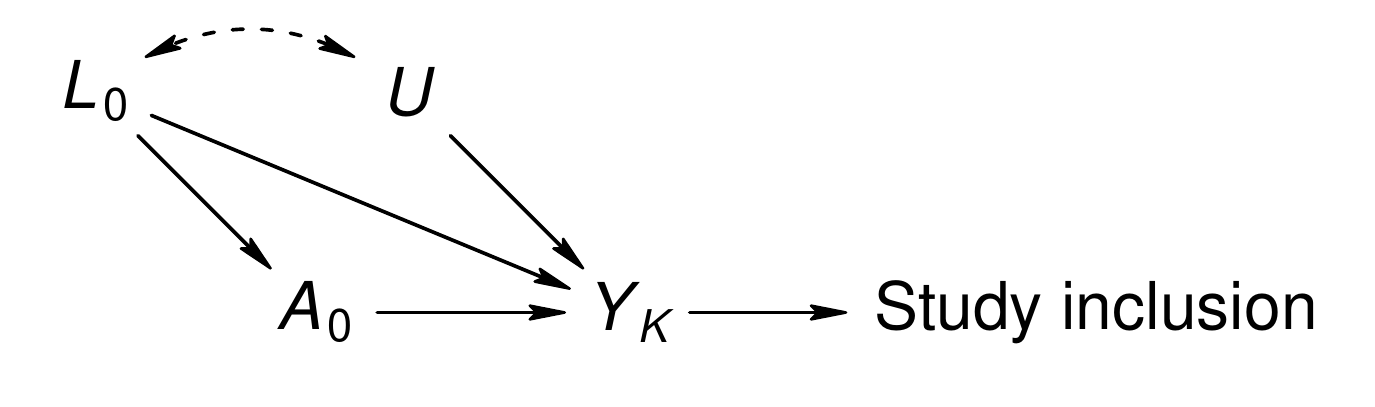}
\end{figure}

\subsection{Risk-set sampling for intention-to-treat effect}
With risk-set (or incidence density) sampling, %study participants may be selected as a control more than once. For 
for all time windows $[t_k,t_{k+1})$, $k=0,...,K-1$, every subject who is event-free at $t_k$ is eligible for selection as a control for the period $[t_k,t_{k+1})$. This means that study participants may be selected as a control more
than once.

Consider the intention-to-treat effect quantified by the marginal (discrete-time) hazard ratio (or rate ratio)
\begin{align*}
\frac{\Pr(Y_{k+1}(1)=1|Y_k(1)=0)}{\Pr(Y_{k+1}(0)=1|Y_k(0)=0)}.
\end{align*}
(For identification of a conditional hazard ratio, see Theorem~\ref{th:DensitySamplingConditionalITT}, Supplementary Appendix~\ref{app:noMatching}.)
For identification of the above marginal hazard ratio under risk-set sampling, it is assumed that selection as a control between $t_k$ and $t_{k+1}$, $S_k$, is independent of the baseline covariates and exposure given eligibility at $t_k$ (i.e., $Y_k=0$). It is also assumed that the sampling probability among those eligible, $\Pr(S_k=1|Y_k=0)$, is constant across time windows $k=0,...,K-1$. %To ensure that $\Pr(S_k=1|Y_k=0)$ remains constant, 
To this end, it suffices that the marginal hazard $\Pr(Y_{k+1}=1|Y_k=0)$ remains constant across time windows and that every $k$th sampling fraction $\Pr(S_k=1)$ is equal, up to a proportionality constant, to the probability $\Pr(Y_{k+1}=1,Y_k=0)$ of an incident case in the $k$th window (see remark to Theorem~\ref{th:DensitySamplingMarginalITT}, Supplementary Appendix~\ref{app:noMatching}). 
For practical purposes, this suggests sampling a fixed number of controls for every case from among the set of eligible individuals. To illustrate, consider Figure~\ref{fig:setupIllustration} and note first of all that the individual represented by row 1 trivially qualifies as a case, because the individual survived until the event occurred. Because the event was sustained between $t_5$ and $t_6$, the proposed sampling suggests selecting a fixed number of controls from among those who are eligible at $t_5$. Thus, rows (and only rows) 4 through 9 as well as row 1 itself in Figure~\ref{fig:setupIllustration} qualify for selection as a control for this case. Even though the individual of row 1 is a case, the individual may also be selected as a control when the individuals of row 2, 3 and 6 (but not 8) sustain the event.

Once cases and controls are selected, we can start to derive inverse probability weights $W$ according to equation \eqref{eq:weights1}. We then compute the odds of baseline exposure among cases in the pseudopopulation that is obtained by weighting everyone by $W$ and the odds of baseline exposure among controls weighted by $W$ multiplied by the number of times the individual was selected as a control. The ratio of these odds coincides with the target hazard ratio under the three key identifiability conditions of consistency, baseline conditional exhangeability and positivity together with the assumption that the hazards in the numerator and denominator of the causal hazard ratio are constant across the time windows.

\subsection{Risk-set sampling for per-protocol effect}
For the per-protocol effect quantified by the (discrete-time) hazard ratio (or rate ratio)
\begin{align*}
\frac{\Pr(Y_{k+1}(\overline{1})=1|Y_k(\overline{1})=0)}{\Pr(Y_{j+1}(\overline{0})=1|Y_{k}(\overline{0})=0)},
\end{align*}
eligibility again requires that the respective subject is event-free at $t_k$ (i.e., $Y_k=0$). %For example, it means that the individual represented by row 1 of Figure~\ref{fig:setupIllustration}, is eligible for selection as a control between $t_5$ and $t_6$, even though $A_5$ differs from $A_0$. This individual provides valuable information about the propensity to deviate from a protocol, but is later censored when we compute the odds ratio.
Selection as a control between $t_k$ and $t_{k+1}$, $S_k$, is further assumed independent of covariate and exposure history up to $t_k$ given eligibility at $t_k$ (but see Supplementary Appendix~\ref{app:noMatching} for a slightly weaker assumption). As for the intention-to-treat effect, it is also assumed that the probability to be selected as a control $S_k$ given eligibility is constant across time windows. This assumption is guaranteed to hold if the marginal hazard $\Pr(Y_{k+1}=1|Y_k=0)$ remains constant across time windows and that every $k$th sampling fraction $\Pr(S_k=1)$ is equal, up to a proportionality constant, to the probability of an incident case in the $k$th window. Figure~\ref{fig:setupIllustration} shows five incident events yet only three qualify as a case (rows 2, 3 and 8) when it concerns per-protocol effects. When the first case emerges (row 2), all rows meet the eligibility criterion for selection as a control. When the second emerges, the individual of row 2, who fails to survive event-free until $t_4$, is precluded as a control. When the case of row 8 emerges, only the individuals of rows 4, 5, 7 and 9 are eligible as controls. %We find that eligibility is somehow monotone: once ineligible, always ineligible.

Once cases and controls are selected, we can start to derive time-varying inverse probability weights according to 
\begin{align*}
W_k&=\prod_{j=0}^k\Bigg[\frac{A_j}{\Pr(A_j=1|L_0,...,L_{j},A_0,...,A_{j-1},Y_j=0,S_j=1)}\\&\qquad\qquad{}+{}\frac{1-A_j}{1-\Pr(A_j=1|L_0,...,L_{j},A_0,...,A_{j-1},Y_j=0,S_j=1)}\Bigg].
\end{align*}
It is important to note that the weights are derived from control information but are nonetheless used to weight both cases and controls \citep{Robins1999}. The denominators of the weights describe the propensity to switch exposure level. However, once the weights are derived, every subject is censored from the time that they fail to adhere to one of the protocols of interest for all downstream analysis. The uncensored exposure levels are therefore constant over time. We then compute the baseline exposure odds among cases, weighted by the weights $W_k$ corresponding to the interval $[t_k,t_{k+1})$ of the incident event (i.e., $Y_{k}=0, Y_{k+1}=0$), as well as the baseline exposure odds among controls, weighted by $\sum_{k=0}^{K-1}W_kS_k$, the weighted number of times selected as control. The ratio of these odds equals the target hazard ratio under the three key identifiability conditions of consistency, sequential conditional exchangeability, and positivity together with the assumption that hazards in the numerator and denominator of the causal hazard ratio for the per-protocol effect are constant across the time windows.

\section{Case-control studies with matching}
\TableMatching{} gives an overview of identification results for case-control studies with exact pair matching. Formal statements and proofs are given in Supplementary Appendix~\ref{app:MatchingExact}, which also includes a generalisation of the results of \TableMatching{} to exact 1-to-$M$ matching. While the focus in this section is on exact covariate matching, for partial matching we refer the reader to Supplementary Appendix~\ref{app:MatchingPartial}, where we consider parametric identification by way of conditional logistic regression.

Pair matching involves assigning a single control exposure level, which we denote by $A'$, to every case. As for case-control studies without matching, in a case-control studies with matching an individual qualifies as a case if and only if they sustain the event of interest by the administrative study end (i.e., $Y_K=1$) and adhered to one of the protocols of interest until the time of the incident event. How a matched control exposure is assigned is encoded in the sampling scheme and the assumptions of \TableMatching{}. For example, for identification of the causal marginal risk ratio under case-base sampling,  $A'$ is sampled from all study participants whose baseline covariate value matches that of the case, independently of the participants' baseline exposure value and whether they survive until the end of study. The matching is exact in the sense that the control exposure information is derived from an individual who has the same value for the baseline covariate as the case.

The identification strategy is the same for all results listed in \TableMatching{}. Only the case-control pairs $(A_0,A')$ with discordant exposure values (i.e., $(1,0)$ or $(0,1)$) are used. Under the stated sampling schemes and assumptions, the respective estimands are identified by the ratio of discordant pairs.

\section{Discussion} 
This paper gives a formal account of how and when causal effects can be identified in case-control studies and, as such, underpins the case-control application of \citeauthor{Dickerman2020} (\citeyear{Dickerman2020}).
Like \citeauthor{Dickerman2020}, we believe that case-control studies should generally be regarded as being nested within cohort studies. This view emphasises that the threats to the validity of cohort studies should also be considered in case-control studies. For example, in case-control applications with risk-set sampling, researchers often consider the covariate and exposure status only at, or just before, the time of the event (for cases) or the time of sampling (for controls). However, where a cohort study would require information on baseline levels or the complete treatment and covariate history of participants, one should suspect that this holds for the nested case-control study too. To gain clarity, we encourage researchers to move away from using person-years, -weeks, or -days (rather than individuals) as the default units of inference \citep{Hernan2015}, and to realise that inadequately addressed deviations from a target trial may lead to bias (or departure from identifiability), regardless of whether the study that attempts to emulate it is a case-control or a cohort study \citep{Dickerman2020}.

%\newpage
\pagebreak
%\SingleSpacing
\global\pdfpageattr\expandafter{\the\pdfpageattr/Rotate 90} % to rotate page in viewer

\rotatebox{90}{\hspace{3pt}\begin{minipage}{\textheight}
\textbf{\TableMatching}. Overview of (non-parametric) identification results for case-control studies with exact pair matching.\end{minipage}}
\rotatebox{90}{
\begin{tabularx}{\textheight}{YLZL}
Sampling scheme & Estimand & Assumptions & Identification strategy \\ \midrule
% %%%%%%%%%%%%%%%%%%%
% % Case-base
% %%%%%%%%%%%%%%%%%%%
Case-base & Risk ratio for intention-to-treat effect

\medskip\centering{$\displaystyle\frac{\Pr(Y_K(1)=1)}{\Pr(Y_K(0)=1)}$} & {\Item {Matched control exposure $A'$ sampled from the baseline exposure levels of all subjects with same baseline covariate level $L_0$ as case, independently of the subjects' baseline exposure or survival status}
\newline\Item Consistency\newline\Item Baseline conditional exchangeability\newline\Item Positivity \newline\Item $\Pr(Y_K=1|L_0=l,A_0=1)/\Pr(Y_K=1|L_0=l,A_0=0)$ constant across levels $l$} \newline(Theorem~\ref{th:CaseCohortMatchingITT})  & 1. Compute the frequency of discordant case-control pairs with $A_0=1$ and $A'=0$\newline2. Compute the frequency of discordant case-control pairs with $A_0=0$ and $A'=1$ \newline3. Take the ratio of the results of steps 1 and 2
\\\cmidrule{1-4}
% %%%%%%%%%%%%%%%%%%%
% % Survivor
% %%%%%%%%%%%%%%%%%%%
Survivor& Odds ratio for intention-to-treat effect 

\medskip\centering{$\displaystyle\frac{\mathrm{Odds}(Y_K(1)=1|L_0)}{\mathrm{Odds}(Y_K(0)=1|L_0)}$} & {\Item {Matched control exposure $A'$ sampled from all the baseline exposure levels of all survivors ($Y_K=0$) with same value for $L_0$ as case, independently of the subjects' baseline exposure} \newline\Item Consistency\newline\Item Baseline conditional exchangeability\newline\Item Positivity \newline\Item $\mathrm{Odds}(Y_K=1|L_0,A_0=1)/\mathrm{Odds}(Y_K=1|L_0,A_0=0)$ constant across levels $l$} \newline (Theorem~\ref{th:SurvivorSamplingMatchingITT}) & (Same as identification strategy for case-base sampling)
\\%\cmidrule{1-4}
\end{tabularx}
}
\rotatebox{90}{\hspace{3pt}\begin{minipage}{\textheight}
See text or Supplementary Material for elaboration on assumptions.\end{minipage}}

\newpage

\rotatebox{90}{\hspace{3pt}\begin{minipage}{\textheight}
\textbf{\TableMatching} (continued).\end{minipage}}
\rotatebox{90}{
\begin{tabularx}{\textheight}{YZZY}
% %%%%%%%%%%%%%%%%%%%
% % Risk-set ITT
% %%%%%%%%%%%%%%%%%%%
Sampling scheme & Estimand & Assumptions & Identification strategy \\ \midrule
 Risk-set &  Hazard ratio for intention-to-treat effect 

\medskip\centering{$\displaystyle\frac{\Pr(Y_{k+1}(1)=1|L_0,Y_{k}(1)=0)}{\Pr(Y_{k+1}(0)=1|L_0,Y_{k}(0)=0)}$}
& {\Item {For a case with incident event in $[t_k,t_{k+1})$ (i.e., $Y_k=0,Y_{k+1}=1$), matched control exposure $A'$ sampled from the baseline exposure levels of all subjects that are event-free at $t_k$ ($Y_k=0$) and have the same value for $L_0$ as case. Sampling among these individuals is independent of baseline exposure or survival status}
\newline\Item Consistency\newline\Item Baseline conditional exchangeability\newline\Item Positivity \newline\Item ${\Pr(Y_{k+1}=1|L_0=l,A_0=1,Y_k=0)}/\newline{\Pr(Y_{k+1}=1|L_0=l,A_0=0,Y_k=0)}$ constant across levels $k,l$} \newline (Theorem~\ref{th:DensitySamplingMatchingITT})  &  (Same as identification strategy for case-base sampling)
\\\cmidrule{2-4}
% %%%%%%%%%%%%%%%%%%%
% % Risk-set PP
% %%%%%%%%%%%%%%%%%%%
 &  Hazard ratio for per-protocol effect 

\medskip\centering{$\displaystyle\frac{\Pr(Y_{k+1}(\overline{1})=1|L_0,...,L_k,A_0=...=A_k=1,Y_{k}(\overline{1})=0)}{\Pr(Y_{k+1}(\overline{0})=1|L_0,...,L_k,A_0=...=A_k=0,Y_{k}(\overline{0})=0)}$}
& {\Item {For a case with incident event in $[t+k,t_{k+1})$ (i.e., $Y_k=0,Y_{k+1}=1$), matched control exposure $A'$ sampled from the baseline exposure levels $A_0$ of all individuals who adhered to one of the protocols until $t_k$ (i.e., $A_0=...=A_k$) and have covariate history up to $t_k$. Sampling among these individuals is independent of baseline exposure or survival status}
\newline\Item Consistency\newline\Item Positivity \newline\Item ${\Pr(Y_{k+1}=1|L_0,...,L_k,A_0=...=A_k=1,Y_k=0)}/\newline{\Pr(Y_{k+1}=1|L_0,...,L_k,A_0=...=A_k=0,Y_k=0)}$ constant across levels $k$ and independent of $L_0,...,L_k$}\newline (Theorem~\ref{th:DensitySamplingMatchingPP})  &  (Same as identification strategy for case-base sampling)\\
\end{tabularx}
}
\rotatebox{90}{\hspace{3pt}\begin{minipage}{\textheight}
%See text or Supplementary Material for elaboration on assumptions.
\end{minipage}}

\pagebreak
\global\pdfpageattr\expandafter{\the\pdfpageattr/Rotate 0}

What is meant by a cohort study differs between authors and contexts \citep{Vandenbroucke2012}. The term `cohort' may refer to either a `dynamic population', or a `fixed cohort', whose ``membership is defined in a permanent fashion'' and ``determined by a single defining event and so becomes permanent'' \citep{Rothman2008}. While it may sometimes be of interest to ask what would have happened with a dynamic cohort (e.g., the residents of a country) had it been subjected to one treatment protocol versus another, the results in this paper relate to fixed cohorts. 

Like the cohort studies within which they are (at least conceptually) nested, case-control studies require an explicit definition of time zero, the time at which a choice is to be made between treatment strategies or protocols of interest \citep{Dickerman2020}. Given a fixed cohort, time zero is generally determined by the defining event of the cohort (e.g., first diagnosis of a particular disease or having survived one year since diagnosis). This event may occur at different calendar times for different individuals. However, while a fixed cohort may be `open' to new members relative to calendar time, it is always `closed' along the time axis on which all subject-specific time zero's take a common point.

In this paper, time was regarded as discrete. Since we considered arbitrary intervals between time points and because, in real-world studies, time is never measured in a truly continuous fashion, this does not represent an important limitation for practical purposes. It is however important to note that the intervals between interventions and outcome assessments (in a target trial) are an intrinsic part of the estimand that lies at the start of investigation. Careful consideration of time intervals in the design of the conceptual target trial and of the actual cohort or case-control study is therefore warranted.

We emphasize that identification and estimation are distinct steps in causal inference. Although our focus was on the former, identifiability expressions often naturally translate into estimators. The task of finding the estimator with the most appealing statistical properties is not necessarily straightforward, however, and is beyond the scope of this paper.

We specifically studied two causal contrasts (i.e., pairs of interventions), one corresponding to intention-to-treat effects and the other to always-versus-never per-protocol effects of a time-varying exposure. There are of course many more causal contrasts, treatment regimes and estimands conceivable that could be of interest. We argue that also for these estimands, researchers should seek to establish identifiability before they select an estimator.

The conditions under which identifiability is to be sought for practical purposes may well include more constraints or obstacles to causal inference, such as additional missingness (e.g., outcome censoring) and measurement error, than we have considered here. While some of our results assume that hazards or hazard ratios remain constant over time, in many cases these are likely time-varying \citep{Lefebvre2006,Guess2006}. There are also more case-control designs (e.g., the case-crossover design) to consider. These additional complexities and designs are beyond the scope of this paper and represent an interesting direction for future research.

The case-control family of study designs is an important yet often misunderstood tool for identifying causal relations \citep{Knol2008,Pearce2016,Mansournia2018,Labrecque2021}. 
Although there is much to be learned, we believe that the modern arsenal for causal inference, which includes counterfactual thinking, is well-suited to make transparent for these classical epidemiological study designs what assumptions are sufficient or necessary to endow the study results with a causal interpretation and, in turn, help resolve or prevent misunderstanding.

\section*{Conflicts of interest} None declared.

\section*{Sources of funding}
RHHG was funded by the Netherlands Organization for Scientific Research (NWO-Vidi project 917.16.430).
The content is solely the responsibility of the authors and does not necessarily represent the official views of the funding bodies.

%\bibliography{References}

\newpage
%\DoubleSpacing

\setcounter{page}{1}

\section*{\huge\scshape{Supplementary material to\\ `\mytitle'}}

%{\color{red} Theorem [regarding density sampling]: For small interval lengths: sampling c control individuals for every incident case is approximately equivalent to sampling c' control individuals for every at-risk person at the start of each interval}

\section*{Table of contents}
\newcommand{\dent}{\hspace{10pt}}
\newcommand{\tocEntry}[2]{\noindent\begin{minipage}[t]{10pt}~\end{minipage}\begin{minipage}[t]{.7\linewidth}\undent #1\end{minipage}\hfill\begin{minipage}[t]{.2\linewidth}\begin{flushright}#2\end{flushright}\end{minipage}}
\newcommand{\tocEntryItalic}[2]{\noindent\begin{minipage}[t]{20pt}~\end{minipage}\begin{minipage}[t]{.7\linewidth}\undent \textit{#1}\end{minipage}\hfill\begin{minipage}[t]{.2\linewidth}\begin{flushright}\textit{#2}\end{flushright}\end{minipage}}

\newcommand{\titleBaseSamplingMarginalITT}{Case-base sampling for marginal intention-to-treat effect}
\newcommand{\titleBaseSamplingConditionalITT}{Case-base sampling for conditional intention-to-treat effect}
\newcommand{\titleSurvivorSamplingConditionalITT}{Survivor sampling for conditional intention-to-treat effect}
\newcommand{\titleDensitySamplingMarginalITT}{Risk-set sampling for marginal intention-to-treat effect}
\newcommand{\titleDensitySamplingConditionalITT}{Risk-set sampling for conditional intention-to-treat effect}
\newcommand{\titleDensitySamplingPP}{Risk-set sampling for marginal per-protocol effect}
\newcommand{\titleCaseCohortMatchingITT}{Case-base sampling for marginal intention-to-treat effect}
\newcommand{\titleSurvivorSamplingMatchingITT}{Survivor sampling for conditional intention-to-treat effect}
\newcommand{\titleDensitySamplingMatchingITT}{Risk-set sampling for conditional intention-to-treat effect}
\newcommand{\titleDensitySamplingMatchingPP}{Risk-set sampling for conditional per-protocol effect}
\newcommand{\titleConditionalLogisticRegression}{Conditional logistic regression for conditional intention-to-treat effect}

\tocEntry{}{Page}\\
\tocEntry{Appendix~\ref{app:notation}: Notation and set-up}{\pageref{app:notation}}\\
\tocEntry{Appendix~\ref{app:noMatching}: Identification results for non-matching strategies}{\pageref{app:noMatching}}\\
\tocEntryItalic{Theorem~\ref{th:BaseSamplingMarginalITT}: \titleBaseSamplingMarginalITT}{\pageref{th:BaseSamplingMarginalITT}}\\
\tocEntryItalic{Theorem~\ref{th:BaseSamplingConditionalITT}: \titleBaseSamplingConditionalITT}{\pageref{th:BaseSamplingConditionalITT}}\\
\tocEntryItalic{Theorem~\ref{th:SurvivorSamplingConditionalITT}: \titleSurvivorSamplingConditionalITT}{\pageref{th:SurvivorSamplingConditionalITT}}\\
\tocEntryItalic{Theorem~\ref{th:DensitySamplingMarginalITT}: \titleDensitySamplingMarginalITT}{\pageref{th:DensitySamplingMarginalITT}}\\
\tocEntryItalic{Theorem~\ref{th:DensitySamplingConditionalITT}: \titleDensitySamplingConditionalITT}{\pageref{th:DensitySamplingConditionalITT}}\\
\tocEntryItalic{Theorem~\ref{th:DensitySamplingPP}: \titleDensitySamplingPP}{\pageref{th:DensitySamplingPP}}\\
\tocEntry{Appendix~\ref{app:MatchingExact}: Identification results for exact 1:$M$ matching strategies}{\pageref{app:MatchingExact}}\\
\tocEntryItalic{Theorem~\ref{th:CaseCohortMatchingITT}: \titleCaseCohortMatchingITT}{\pageref{th:CaseCohortMatchingITT}}\\
\tocEntryItalic{Theorem~\ref{th:SurvivorSamplingMatchingITT}: \titleSurvivorSamplingMatchingITT}{\pageref{th:SurvivorSamplingMatchingITT}}\\
\tocEntryItalic{Theorem~\ref{th:DensitySamplingMatchingITT}: \titleDensitySamplingMatchingITT}{\pageref{th:DensitySamplingMatchingITT}}\\
\tocEntryItalic{Theorem~\ref{th:DensitySamplingMatchingPP}: \titleDensitySamplingMatchingPP}{\pageref{th:DensitySamplingMatchingPP}}\\
\tocEntry{Appendix~\ref{app:MatchingPartial}: Parametric identification by conditional logistic regression for exact or partial 1:$M$ matching}{~\\\pageref{app:MatchingPartial}}\\\newline
\tocEntryItalic{Theorem~\ref{th:ConditionalLogisticRegression}: \titleConditionalLogisticRegression}{\\\pageref{th:ConditionalLogisticRegression}}\\\newline

%\newpage
\begin{myAppendix}[: Notation and set-up]\label{app:notation}

We will suppose that the interest lies with the effect of a time-varying exposure that can take one of two levels at any given time on a failure time outcome. In particular, we consider a strictly increasing sequence $(t_0,t_1,...,t_K)$ of $K+1$ time points (with $t_{K+1}=-t_{-1}=+\infty$ for notational convenience). For $k=0,1,...,K-1$, let $A_k$ denote the level of time-varying exposure of interest at $t_k$. We denote the history of any stochastic sequence $(X_0,X_1,...,X_{K-1})$ up to and including $t_k$ by $\overline{X}_k=(X_0,X_1...,X_k)$ for $k=0,1,...,K-1$ (and let $\overline{X}=\overline{X}_{K-1}$ and $\overline{X}_{-1}=0$ for notational convenience). For example, $\overline{A}=(A_0,A_1,...,A_{K-1})$. Denote by $T(\overline{a})$ the counterfactual time elapsed until the event of interest since $t_0$ that would have been realised had $\overline{A}$ been set to $\overline{a}$, and let $Y_k(\overline{a})=I(T(\overline{a})<t_k)$ for $k=0,1,...,K$, where $I$ represents the indicator function. By convention, we stipulate that for all $k$, $Y_k(\overline{a})$ is invariant to the $k$th through $K-1$th elements of $\overline{a}$ (i.e., current survival status is not affected by future exposures). With slight abuse of notation, for $k=0,1...,K$, we let $Y_k(a_0)$ denote the outcome that would have been realised had (only) $A_0$ been set to $a_0$.

% consistency
\subsection*{Consistency} 
For theorems about per-protocol effects, we assume consistency of the form: for $k=1,...,K$ and all $\overline{a}$, $Y_k(\overline{a})=Y_k$ if $a_{l}=A_{l}$ for all $l=0,...,k-1$ such that $Y_l=0$. For theorems about intention-to-treat effects, a weaker condition is sufficient and assumed: for $k=1,...,K$ and $a=0,1$, $Y_k(a)=Y_k$ if $a=A_0$. The assumption may be further relaxed for theorems in which the estimand does not involve $Y_k(a)$, $k<K$: for $a=0,1$, $Y_K(a)=Y_K$ if $a=A_0$.

% conditional exchangeability
\subsection*{Conditional exchangeability}
We also consider a sequence of variables $\overline{L}=(L_0,L_1,...,L_{K-1})$ that satisfies one of the following conditions:
\begin{align}
\forall{k},\forall{\overline{a}}:(Y_{k+1}(\overline{a}),...,Y_{K}(\overline{a}))\CI A_k|Y_k(\overline{a})=0,\overline{L}_k,\overline{A}_{k-1}=\overline{a}_{k-1}, \tag{sequential conditional exchangeability, SCE}
\end{align}
where $\overline{a}_{k-1}$ is understood to represent the $(k-1)$th through $(K-1)$th elements of $\overline{a}$,
or
\begin{align}
\forall{a_0}:(Y_{1}(a_0),...,Y_{K}(a_0))\CI A_0|L_0, \tag{baseline conditional exchangeability, BCE}
\end{align}
although sometimes a weaker form of BCE suffices: $\forall{a_0}:Y_{K}(a_0)\CI A_0|L_0$.

% positivity
\subsection*{Positivity}
For the theorems that follow, we assume positivity to preclude division by zero and undefined conditional probabilities, so that the weights that we will encounter are finite and strictly greater than 1. The assumption can sometimes be relaxed if we are willing to interpolate or extrapolate under (parametric) modelling assumptions.
\end{myAppendix}

\begin{myAppendix}[: Identification results for non-matching strategies]\label{app:noMatching}
\subsection*{Intention-to-treat effect}
% In this subsection, an individual qualifies as a case if and only if $Y_K=1$. Let $S$ be an indicator variable for selection as a control and $S_k$ for selection as a control for the period $[t_k,t_{k+1})$ specifically, $k\ge 0$. In the theorems below, we assume that
% \begin{align*}
% \Pr(S=1|L_0,A_0) = \Pr(S=1) = \delta \tag{\BaseSamplingNoMatching{}}
% \end{align*}
% or
% \begin{align*}
% %\Pr(S_{k}=1|\overline{L},\overline{A},\overline{Y}_K,S_{k-1}) &= 
% \Pr(S=1|L_0,A_0,Y_K) = \Pr(S=1|L_0,Y_K) &= 
% \delta_{L_0}\times (1-Y_{K}) \tag{\SurvivorSamplingNoMatching{}}
% \end{align*}
% or
% \begin{align*}
% \Pr(S_k=1|L_0,A_0,Y_{k}) = \Pr(S_k=1|Y_{k}) &= 
% \delta\times(1-Y_k), \tag{\RiskSetSamplingNoMatching{}}
% \end{align*}
% for some $\delta,\delta_{L_0}\in (0,1]$.

For simplicity, it is assumed below that the covariates are discrete. The results can however be extended to more general distributions.

% Theorem (or two theorems? Nahh.. corollary): Case-base (with/without weighting)
% Theorem: Risk-set sampling (with/without weighting)
% Theorem: Survivor sampling (without weighting only)

\begin{theorem}[\titleBaseSamplingMarginalITT]\label{th:BaseSamplingMarginalITT}
Suppose BCE holds as well as
\begin{align*}
\Pr(S=1|L_0,A_0) = \Pr(S=1) = \delta \tag{\BaseSamplingNoMatching{}}
\end{align*}
for some $\delta\in(0,1]$. %Also assume positivity of the form: $0<\Pr(A_0=1|L_0,Y_K=1),\Pr(A_0=1|L_0,S=1)<1$. 
Then,
\begin{align*}
\frac{\displaystyle\frac{\Exp\big[I(A_0=1)W|Y_K=1\big]}{\Exp\big[I(A_0=0)W|Y_K=1\big]}}{\displaystyle
\frac{\Exp\big[I(A_0=1)W|S=1\big]}{\Exp\big[I(A_0=0)W|S=1\big]}
} &= \frac{\Pr(Y_{K}(1)=1)}{\Pr(Y_{K}(0)=1)},
\end{align*}
where 
\begin{align*}
W=\frac{1}{\Pr(A_0=a|L_0,S=1)}\bigg|_{a=A_0},
\end{align*}
\end{theorem}
\begin{proof}
First, observe that $\Pr(A_0=a|L_0,S=1)=\Pr(A_0=a|L_0)$ for $a=0,1$, because
\begin{align*}
\Pr(A_0=a|L_0,S=1)&=\frac{\Pr(S=1|L_0,A_0=a)\Pr(A_0=a|L_0)}{\Pr(S=1|L_0)}\\
&=\frac{\delta}{\delta}\Pr(A_0=a|L_0) \tag{by \BaseSamplingNoMatching{}}\\
&=\Pr(A_0=a|L_0)
\end{align*}
Hence,
\begin{align*}
W=
\frac{1}{\Pr(A_0=a|L_0)}\bigg|_{a=A_0}.
\end{align*}

Now, consider the numerator of the left-hand side of the main equation in Theorem~\ref{th:BaseSamplingMarginalITT} and note that, because of the above, we have
\begin{align*}
\frac{\Exp\big[I(A_0=1)W|Y_K=1\big]}{\Exp\big[I(A_0=0)W|Y_K=1\big]}
&= \frac{\sum_{y=0}^1\Exp\big[I(A_0=1)WY_K|Y_K=y\big]\Pr(Y_K=y)}{\sum_{y=0}^1\Exp\big[I(A_0=0)WY_K|Y_K=y\big]\Pr(Y_K=y)}\\
&= \frac{\Exp\big[I(A_0=1)WY_K\big]}{\Exp\big[I(A_0=0)WY_K\big]} \\
&= \frac{\Exp\big[WY_K|A_0=1\big]\Pr(A_0=1)}{\Exp\big[WY_K|A_0=0\big]\Pr(A_0=0)},
\end{align*}
where
\begin{align*}
\Exp\big[WY_K|A_0=a\big] &= \Exp\big\{\Exp\big[WY_K|L_0,A_0=a\big]|A_0=a\big\}\\
&= \sum_{l}\frac{\Pr(Y_K=1|L_0=l,A_0=a)\Pr(L_0=l|A_0=a)}{\Pr(A_0=a|L_0=l)}\\
&= \sum_{l}\frac{\Pr(Y_K(a)=1|L_0=l,A_0=a)\Pr(L_0=l|A_0=a)}{\Pr(A_0=a|L_0=l)} \tag{by consistency}\\
&= \sum_{l}\frac{\Pr(Y_K(a)=1|L_0=l)\Pr(L_0=l|A_0=a)}{\Pr(A_0=a|L_0=l)} \tag{by baseline conditional exchangeability}\\
&= \sum_{l}\frac{\Pr(Y_K(a)=1|L_0=l)\Pr(A_0=a|L_0=l)\Pr(L_0=l)}{\Pr(A_0=a|L_0=l)\Pr(A_0=a)}\\
&= \frac{1}{\Pr(A_0=a)}\sum_{l}\Pr(Y_K(a)=1,L_0=l)\\
&= \frac{\Pr(Y_K(a)=1)}{\Pr(A_0=a)},
\end{align*}
so that
\begin{align*}
\frac{\Exp\big[I(A_0=1)W|Y_K=1\big]}{\Exp\big[I(A_0=0)W|Y_K=1\big]}
&= \frac{\Pr(Y_K(1)=1)}{\Pr(Y_K(0)=1)}.
\end{align*}

Next, consider the denominator of the left-hand side of the main equation in Theorem~\ref{th:BaseSamplingMarginalITT} and observe that
\begin{align*}
\frac{\Exp\big[I(A_0=1)W|S=1\big]}{\Exp\big[I(A_0=0)W|S=1\big]} = \frac{\Exp\big[I(A_0=1)WS\big]}{\Exp\big[I(A_0=0)WS\big]} &= \frac{\Exp\big[WS|A_0=1\big]\Pr(A_0=1)}{\Exp\big[WS|A_0=0\big]\Pr(A_0=0)},
\end{align*}
where
\begin{align*}
\Exp\big[WS|A_0=a\big] &= \Exp\{\Exp\big[WS|L_0,A_0=a\big]|A_0=a\}\\
&= \sum_l\frac{\Pr(S=1|L_0,A_0=a)\Pr(L_0=l|A_0=a)}{\Pr(A_0=a|L_0=l)}\\
&= \sum_l\frac{\delta\Pr(L_0=l|A_0=a)}{\Pr(A_0=a|L_0=l)} \tag{by \BaseSamplingNoMatching{}}\\
&= \frac{\delta}{\Pr(A_0=a)}\sum_l\Pr(L_0=l)\\
&= \frac{\delta}{\Pr(A_0=a)},
\end{align*}
so that 
\begin{align*}
\frac{\Exp\big[I(A_0=1)W|S=1\big]}{\Exp\big[I(A_0=0)W|S=1\big]} &= 1.
\end{align*}
It follows that \begin{align*}
\frac{\displaystyle\frac{\Exp\big[I(A_0=1)W|Y_K=1\big]}{\Exp\big[I(A_0=0)W|Y_K=1\big]}}{\displaystyle
\frac{\Exp\big[I(A_0=1)W|S=1\big]}{\Exp\big[I(A_0=0)W|S=1\big]}
} &= \frac{\Pr(Y_{K}(1)=1)}{\Pr(Y_{K}(0)=1)}.
\end{align*}
\end{proof}

\begin{theorem}[\titleBaseSamplingConditionalITT]\label{th:BaseSamplingConditionalITT}
Suppose BCE hold as well as \BaseSamplingNoMatching{}, or the weaker version $\Pr(S=1|L_0,A_0)=\Pr(S=1|L_0)=\delta_{L_0}\in(0,1]$. %Also assume positivity of the form: $0<\Pr(A_0=1|L_0,Y_K=1),\Pr(A_0=1|L_0,S=1)<1$. 
Then,
\begin{align*}
\frac{\displaystyle\frac{\Exp\big[I(A_0=1)|L_0,Y_K=1\big]}{\Exp\big[I(A_0=0)|L_0,Y_K=1\big]}}{\displaystyle
\frac{\Exp\big[I(A_0=1)|L_0,S=1\big]}{\Exp\big[I(A_0=0)|L_0,S=1\big]}
} &= \frac{\Pr(Y_{K}(1)=1|L_0)}{\Pr(Y_{K}(0)=1|L_0)}.
\end{align*}
\end{theorem}
%The proof to Theorem~\ref{th:BaseSamplingConditionalITT} is similar to that of Theorem~\ref{th:BaseSamplingMarginalITT} and therefore omitted.

\begin{proof}
We have
\begin{align*}
\frac{\Exp\big[I(A_0=1)|L_0,Y_K=1\big]}{\Exp\big[I(A_0=0)|L_0,Y_K=1\big]}
&= \frac{\sum_{y=0}^1\Exp\big[I(A_0=1)Y_K|L_0,Y_K=y\big]\Pr(Y_K=y|L_0)}{\sum_{y=0}^1\Exp\big[I(A_0=0)Y_K|L_0,Y_K=y\big]\Pr(Y_K=y|L_0)}\\
&= \frac{\Exp\big[I(A_0=1)Y_K|L_0\big]}{\Exp\big[I(A_0=0)Y_K|L_0\big]} \\
&= \frac{\Exp\big[Y_K|L_0,A_0=1\big]\Pr(A_0=1|L_0)}{\Exp\big[Y_K|L_0,A_0=0\big]\Pr(A_0=0|L_0)} \\
&= \frac{\Exp\big[Y_K(1)|L_0,A_0=1\big]\Pr(A_0=1|L_0)}{\Exp\big[Y_K(0)|L_0,A_0=0\big]\Pr(A_0=0|L_0)} \tag{by consistency}\\
&= \frac{\Exp\big[Y_K(1)|L_0\big]\Pr(A_0=1|L_0)}{\Exp\big[Y_K(0)|L_0\big]\Pr(A_0=0|L_0)} \tag{by baseline conditional exchangeability}.
\end{align*}
Also,
\begin{align*}
\frac{\Exp\big[I(A_0=1)|L_0,S=1\big]}{\Exp\big[I(A_0=0)|L_0,S=1\big]} &= \frac{\Exp\big[I(A_0=1)S|L_0\big]}{\Exp\big[I(A_0=0)S|L_0\big]} \\&= \frac{\Exp\big[S|L_0,A_0=1\big]\Pr(A_0=1|L_0)}{\Exp\big[S|L_0,A_0=0\big]\Pr(A_0=0|L_0)}\\
&=\frac{\delta_{L_0}\Pr(A_0=1|L_0)}{\delta_{L_0}\Pr(A_0=0|L_0)} \tag{under the assumption that $\Pr(S=1|L_0,A_0)=\Pr(S=1|L_0)=\delta_{L_0}\in(0,1]$}\\
&=\frac{\Pr(A_0=1|L_0)}{\Pr(A_0=0|L_0)}.
\end{align*}
It immediately follows that
\begin{align*}
\frac{\displaystyle\frac{\Exp\big[I(A_0=1)|L_0,Y_K=1\big]}{\Exp\big[I(A_0=0)|L_0,Y_K=1\big]}}{\displaystyle
\frac{\Exp\big[I(A_0=1)|L_0,S=1\big]}{\Exp\big[I(A_0=0)|L_0,S=1\big]}
} &= \frac{\Pr(Y_{K}(1)=1|L_0)}{\Pr(Y_{K}(0)=1|L_0)}.
\end{align*}
\end{proof}

\begin{corollary}
If in addition to the conditions of Theorem~\ref{th:BaseSamplingConditionalITT}, \begin{align*}
\frac{\Pr(Y_K=1|L_0=l,A_0=1)}{\Pr(Y_K=1|L_0=l,A_0=0)} = \theta \tag{homogeneity condition \HomogeneousRiskRatios{}}
\end{align*}
for all $l$ and some constant $\theta$, then 
\begin{align*}
\frac{\displaystyle\frac{\Exp\big[I(A_0=1)|L_0,Y_K=1\big]}{\Exp\big[I(A_0=0)|L_0,Y_K=1\big]}}{\displaystyle
\frac{\Exp\big[I(A_0=1)|L_0,S=1\big]}{\Exp\big[I(A_0=0)|L_0,S=1\big]}
}  = \frac{\Pr(Y_K(1)=1)}{\Pr(Y_K(0)=1)},
\end{align*}
because of the collapsibility of the risk ratio.
\end{corollary}

\begin{theorem}[\titleSurvivorSamplingConditionalITT]\label{th:SurvivorSamplingConditionalITT}
Suppose BCE holds as well as \begin{align*}
%\Pr(S_{k}=1|\overline{L},\overline{A},\overline{Y}_K,S_{k-1}) &= 
\Pr(S=1|L_0,A_0,Y_K) = \Pr(S=1|L_0,Y_K) &= 
\delta_{L_0}\times (1-Y_{K}) \tag{\SurvivorSamplingNoMatching{}}
\end{align*}
for some $\delta_{L_0}\in(0,1]$.
Then,
\begin{align*}
\frac{\displaystyle\frac{\Exp\big[I(A_0=1)|L_0,Y_K=1\big]}{\Exp\big[I(A_0=0)|L_0,Y_K=1\big]}}{\displaystyle
\frac{\Exp\big[I(A_0=1)|L_0,S=1\big]}{\Exp\big[I(A_0=0)|L_0,S=1\big]}
} &= \frac{\mathrm{Odds}(Y_{K}(1)=1|L_0)}{\mathrm{Odds}(Y_{K}(0)=1|L_0)}.
\end{align*}
\end{theorem}
\begin{proof}
First, consider the numerator of the left-hand side of the equation in Theorem~\ref{th:SurvivorSamplingConditionalITT} and observe
\begin{align*}
\frac{\Exp\big[I(A_0=1)|L_0,Y_K=1\big]}{\Exp\big[I(A_0=0)|L_0,Y_K=1\big]} &= \frac{\Pr(Y_K=1|L_0,A_0=1)}{\Pr(Y_K=1|L_0,A_0=0)}\mathrm{Odds}(A_0=1|L_0) \\
&= \frac{\Pr(Y_K(1)=1|L_0,A_0=1)}{\Pr(Y_K(1)=1|L_0,A_0=0)}\mathrm{Odds}(A_0=1|L_0) \tag{by consistency}\\
&= \frac{\Pr(Y_K(1)=1|L_0)}{\Pr(Y_K(1)=1|L_0)}\mathrm{Odds}(A_0=1|L_0). \tag{by baseline conditional exchangeability}
\end{align*}

Next, consider the denominator and observe that
\begin{align*}
\frac{\Exp\big[I(A_0=1)|L_0,S=1\big]}{\Exp\big[I(A_0=0)|L_0,S=1\big]} &=
\frac{\Exp\big[I(A_0=1)S|L_0\big]}{\Exp\big[I(A_0=0)S|L_0\big]} \\
&= \frac{\Exp\big[S|L_0,A_0=1\big]}{\Exp\big[S|L_0,A_0=0\big]}\mathrm{Odds}(A_0=1|L_0)\\
&=\frac{\delta_{L_0}\Pr(Y_K=0|L_0,A_0=1)}{\delta_{L_0}\Pr(Y_K=0|L_0,A_0=0)}\mathrm{Odds}(A_0=1|L_0)\tag{by \SurvivorSamplingNoMatching{}}\\
&=\frac{\Pr(Y_K(1)=0|L_0,A_0=1)}{\Pr(Y_K(0)=0|L_0,A_0=0)}\mathrm{Odds}(A_0=1|L_0) \tag{by consistency}\\
&=\frac{\Pr(Y_K(1)=0|L_0)}{\Pr(Y_K(0)=0|L_0)}\mathrm{Odds}(A_0=1|L_0). \tag{by baseline conditional exchangeability}
\end{align*}

It follows that \begin{align*}
\frac{\displaystyle\frac{\Exp\big[I(A_0=1)|L_0,Y_K=1\big]}{\Exp\big[I(A_0=0)|L_0,Y_K=1\big]}}{\displaystyle
\frac{\Exp\big[I(A_0=1)|L_0,S=1\big]}{\Exp\big[I(A_0=0)|L_0,S=1\big]}
} &= \frac{\mathrm{Odds}(Y_{K}(1)=1|L_0)}{\mathrm{Odds}(Y_{K}(0)=1|L_0)}.
\end{align*}
\end{proof}

\begin{remark}[Remark to Theorem~\ref{th:SurvivorSamplingConditionalITT}]
Under BCE, the stronger version of \SurvivorSamplingNoMatching, 
\begin{align*}
\Pr(S=1|L_0,A_0,Y_K) = \Pr(S=1|Y_K) = \delta\times (1-Y_K) \tag{\SurvivorSamplingNoMatching$^\ast$}
\end{align*}
for some $\delta\in(0,1]$ and with \begin{align*}
W &= \frac{1}{\Pr(A_0=a|L_0)}\Bigg|_{a=A_0},
\end{align*}
we have
\begin{align}
\frac{\displaystyle\frac{\Exp\big[I(A_0=1)W|Y_K=1\big]}{\Exp\big[I(A_0=0)W|Y_K=1\big]}}{\displaystyle
\frac{\Exp\big[I(A_0=1)W|S=1\big]}{\Exp\big[I(A_0=0)W|S=1\big]}
} &= \frac{\mathrm{Odds}(Y_{K}(1)=1)}{\mathrm{Odds}(Y_{K}(0)=1)} \label{eq:marginalORunderSurivivalSampling}
\end{align}
(see proof below).
However, from
\begin{align*}
\Pr(A_0=a|L_0,S=1) &= \frac{\Pr(S=1|L_0,A_0=a)\Pr(A_0=a|L_0)}{\Pr(S=1|L_0)}\\
&=\frac{\delta\Pr(Y_K=0|L_0,A_0=a)\Pr(A_0=a|L_0)}{\delta\Pr(Y_K=0|L_0)} \tag{by \SurvivorSamplingNoMatching$^\ast$}\\
&= \Pr(A_0=a|L_0,Y_K=0),
\end{align*}
it follows that the weights $W$ above are not identified by 
\begin{align*}
\frac{1}{\Pr(A_0=a|L_0,S=1)}\Bigg|_{a=A_0}
\end{align*}
when $Y_K\nCI A_0|L_0$. (However, $\Pr(A_0=a|L_0,S=1)$ approximates $\Pr(A_0=a|L_0)$ under a rare event assumption.) In fact, the target marginal odds ratio is not identifiable, under BCE and \SurvivorSamplingNoMatching$^\ast$ with unknown $\delta$, from the available data distribution, which is formed by the distribution of $(L_0,A_0,Y_K,S)|(Y_K=1\vee S=1)$. A proof is given below.

\begin{proof}[Proof of \eqref{eq:marginalORunderSurivivalSampling} under stated conditions]
As shown in the proof to Theorem~\ref{th:BaseSamplingMarginalITT}, \begin{align*}\frac{\Exp\big[I(A_0=1)W|Y_K=1\big]}{\Exp\big[I(A_0=0)W|Y_K=1\big]} &= \frac{\Pr(Y_K(1)=1)}{\Pr(Y_K(0)=1)}.
\end{align*}
Now, \begin{align*}
\frac{\Exp\big[I(A_0=1)W|S=1\big]}{\Exp\big[I(A_0=0)W|S=1\big]} = \frac{\Exp\big[I(A_0=1)WS\big]}{\Exp\big[I(A_0=0)WS\big]} &= \frac{\Exp\big[WS|A_0=1\big]\Pr(A_0=1)}{\Exp\big[WS|A_0=0\big]\Pr(A_0=0)},
\end{align*}
where
\begin{align*}
\Exp\big[WS|A_0=a\big] &= \Exp\{\Exp\big[WS|L_0,A_0=a\big]|A_0=a\}\\
&= \sum_l\frac{\Pr(S=1|L_0,A_0=a)\Pr(L_0=l|A_0=a)}{\Pr(A_0=a|L_0=l)}\\
&= \sum_l\frac{\delta\Pr(Y_K=0|L_0=l,A_0=a)\Pr(L_0=l|A_0=a)}{\Pr(A_0=a|L_0=l)} \tag{by \SurvivorSamplingNoMatching$^\ast$}\\
&= \frac{\delta}{\Pr(A_0=a)}\sum_l\Pr(Y_K=0|L_0=l,A_0=a)\Pr(L_0=l)\\
&= \frac{\delta}{\Pr(A_0=a)}\sum_l\Pr(Y_K(a)=0|L_0=l,A_0=a)\Pr(L_0=l) \tag{by consistency}\\
&= \frac{\delta}{\Pr(A_0=a)}\sum_l\Pr(Y_K(a)=0,L_0=l) \tag{by baseline conditional exchangeability}\\
&= \frac{\delta\Pr(Y_K(a)=0)}{\Pr(A_0=a)},
\end{align*}
so that 
\begin{align*}
\frac{\Exp\big[I(A_0=1)W|S=1\big]}{\Exp\big[I(A_0=0)W|S=1\big]} &= \frac{\Pr(Y_K(1)=0)}{\Pr(Y_K(0)=0)}
\end{align*}
and, in turn,
\begin{align*}
\frac{\displaystyle\frac{\Exp\big[I(A_0=1)W|Y_K=1\big]}{\Exp\big[I(A_0=0)W|Y_K=1\big]}}{\displaystyle
\frac{\Exp\big[I(A_0=1)W|S=1\big]}{\Exp\big[I(A_0=0)W|S=1\big]}
} &= \frac{\mathrm{Odds}(Y_{K}(1)=1)}{\mathrm{Odds}(Y_{K}(0)=1)}.
\end{align*}
\end{proof}

\begin{proof}[Proof of nonidentifiability of target marginal odds ratio under stated conditions]
Consider two distributions of $(L_0,A_0,Y_K,S)$ satisfying \SurvivorSamplingNoMatching$^\ast$, each characterised by the following conditionals:
\begin{align*}
Y_K &\sim \mathrm{Bernoulli}(\alpha),\\
S|Y_K &\sim \mathrm{Bernoulli}(\delta\times (1-Y_K)),\\
L_0|Y_K,S &\sim L_0|Y_K \sim \mathrm{Bernoulli}(5/10-2/10\times Y_K),\\
A_0|L_0,Y_K,S &\sim A_0|L_0,Y_K \sim \mathrm{Bernoulli}(3/10+2/10\times L_0+3/10\times Y_K).
\end{align*}
The parameter values of the distributions are given in the table below.
\begin{center}
\begin{tabular}{ccc}
\toprule
Parameter&Distribution 1&Distribution 2\\
\midrule
$\alpha$&$1/10$&$2/10$\\
$\delta$&$1/10$&$9/40$\\
% $\beta_0$&&\\
% $\beta_1$&&\\
% $\gamma_0$&&\\
% $\gamma_1$&&\\
% $\gamma_2$&&\\
% $\gamma_3$&&\\
\bottomrule
\end{tabular}
\end{center}

Now, for all $l,a,y,s\in\{0,1\}$,
\begin{align*}
&\Pr(L_0=l,A_0=a,Y_K=y,S=s|Y_K=1\vee S=1)\\
&\qquad= \frac{\Pr(L_0=l,A_0=a,Y_K=y,S=s,Y_K=1\vee S=1)}{\Pr(Y_K=1\wedge S=0)+\Pr(Y_K=0\wedge S=1)+\Pr(Y_K=1\wedge S=1)}\\
&\qquad= \frac{I(y=1\vee s=1)\Pr(L_0=l,A_0=a,Y_K=y,S=s)}{\Pr(Y_K=1)+\delta\Pr(Y_K=0)}\\
&\qquad= I(y=1\vee s=1)\frac{\Pr(L_0=l,A_0=a|Y_K=y)\Pr(S=s|Y_K=y)\Pr(Y_K=y)}{\alpha+\delta(1-\alpha)}\\
&\qquad=\left\{\begin{array}{ll}
\displaystyle\Pr(L_0=l,A_0=a|Y_K=0)\bigg(1-\frac{\alpha}{\alpha+\delta(1-\alpha)}\bigg)&\text{if~}y=0\wedge s=1,\\[7pt]
\displaystyle\Pr(L_0=l,A_0=a|Y_K=1)\frac{\alpha}{\alpha+\delta(1-\alpha)}&\text{if~}y=1\wedge s=0,\\[3pt]
% \displaystyle\frac{\Pr(L_0=l,A_0=a|Y_K=1)\Pr(S=1|Y_K=1)\Pr(Y_K=1)}{\alpha+\delta(1-\alpha)}&\text{if~}y=1\wedge s=1\\[3pt]
0&\text{otherwise,}
\end{array}\right.
\end{align*}
where $$\frac{\alpha}{\alpha+\delta(1-\alpha)}=10/19$$ under Distribution 1 and under Distribution 2. Hence, Distribution 1 and 2 imply the same available data distribution. 

However, as we now show, the distributions imply different target marginal odds ratios.
Since
\begin{align*}
\Pr(Y_K(a)=1) &= \sum_{l=0}^1\Pr(Y_K(a)=1|L_0=l)\Pr(L_0=l) \\
&= \sum_{l=0}^1\Pr(Y_K(a)=1|L_0=l,A_0=a)\Pr(L_0=l) \tag{by BCE}\\
&= \sum_{l=0}^1\Pr(Y_K=1|L_0=l,A_0=a)\Pr(L_0=l) \tag{by consistency}\\
&= \sum_{l=0}^1\frac{\Pr(L_0=l,A_0=a|Y_K=1)\Pr(Y_K=1)}{\Pr(L_0=l,A_0=a)}\sum_{y=0}^1\Pr(L_0=l|Y_K=y)\Pr(Y_K=y)\\
&= \sum_{l=0}^1\bigg(1+\frac{\Pr(L_0=l,A_0=a|Y_K=0)\Pr(Y_K=0)}{\Pr(L_0=l,A_0=a|Y_K=1)\Pr(Y_K=1)}\bigg)^{-1}\sum_{y=0}^1\Pr(L_0=l|Y_K=y)\Pr(Y_K=y)%\\
% &= \frac{(1-\beta_0)(1-\alpha)+(1-\beta_0-\beta_1)\alpha}{1+\displaystyle\frac{\gamma_0^a(1-\gamma_0)^{1-a}(1-\beta_0)(1-\alpha)}{(\gamma_0+\gamma_2)^a(1-\gamma_0-\gamma_2)^{1-a}(1-\beta-\beta_1)\alpha}}\\&\qquad{}+\frac{\beta_0(1-\alpha)+(\beta_0+\beta_1)\alpha}{1+\displaystyle\frac{(\gamma_0+\gamma_1)^a(1-\gamma_0-\gamma_1)^{1-a}\beta_0(1-\alpha)}{(\gamma_0+\gamma_1+\gamma_2+\gamma_3)^a(1-\gamma_0-\gamma_1-\gamma_2-\gamma_3)^{1-a}(\beta_0+\beta_1)\alpha}}
\end{align*}
for $a=0,1$, we have
\begin{align*}
\Pr(Y_K(1)=1) &= \frac{5+2\alpha}{10+({25}/{7})/\mathrm{odds}({\alpha})}+\frac{5-2\alpha}{10+(125/12)/\mathrm{odds}(\alpha)}\text{~~and}\\
\Pr(Y_K(0)=1) &= \frac{5+2\alpha}{10+({25}/{2})/\mathrm{odds}({\alpha})}+\frac{5-2\alpha}{10+(125/3)/\mathrm{odds}(\alpha)},
\end{align*}
so that
\begin{align*}
\frac{\mathrm{Odds}(Y_K(1)=1)}{\mathrm{Odds}(Y_K(0)=1)}
&=\left\{\begin{array}{ll}
\displaystyle\frac{587,791}{167,166}\approx3.5&\text{under Distribution 1},\\[7pt]
\displaystyle\frac{512,539}{148,789}\approx 3.4&\text{under Distribution 2}.
\end{array}\right.
\end{align*}

Hence, we found an available data distribution that is compatible with more than one value of the target marginal odds ratio. This concludes the proof.
\end{proof}
\end{remark}

\begin{theorem}[\titleDensitySamplingMarginalITT]\label{th:DensitySamplingMarginalITT}
Suppose BCE holds as well as \begin{align*}
\Pr(S_k=1|L_0,A_0,Y_{k}) = \Pr(S_k=1|Y_{k}) &= 
\delta\times(1-Y_k), \tag{\RiskSetSamplingNoMatching{}}
\end{align*}
for some $\delta\in(0,1]$. %Also assume positivity of the form: {\color{red}$0<...<1$}. 
If \begin{align*}
\Pr(Y_{k+1}(a)=1|Y_k(a)=0)=\theta_{a} \tag{\HomogeneousRates{}}
\end{align*}
for $a=0,1$ and some constants $\theta_0,\theta_1$, then
\begin{align*}
\frac{\displaystyle\frac{\Exp\big[I(A_0=1)W|Y_K=1\big]}{\Exp\big[I(A_0=0)W|Y_K=1\big]}}{\displaystyle
\frac{\Exp\big[I(A_0=1)W\sum_{k=0}^{K-1}S_k\big]}{\Exp\big[I(A_0=0)W\sum_{k=0}^{K-1}S_k\big]}
} &= \frac{\Pr(Y_{k+1}(1)=1|Y_{k+1}(1)=0)}{\Pr(Y_{k+1}(0)=1|Y_{k+1}(0)=0)},
\end{align*}
where 
\begin{align*}
W=\frac{1}{\Pr(A_0=a|L_0,S=1)}\bigg|_{a=A_0},
\end{align*}
\end{theorem}
\begin{proof}
First, observe that $\Pr(A_0=a|L_0,S=1)=\Pr(A_0=a|L_0)$ for $a=0,1$, because
\begin{align*}
\Pr(A_0=a|L_0,S=1)&=\frac{\Pr(S=1|L_0,A_0=a)\Pr(A_0=a|L_0)}{\Pr(S=1|L_0)}\\
&=\frac{\delta}{\delta}\Pr(A_0=a|L_0) \tag{by \RiskSetSamplingNoMatching{}}\\
&=\Pr(A_0=a|L_0)
\end{align*}
Hence,
\begin{align*}
W=
\frac{1}{\Pr(A_0=a|L_0)}\bigg|_{a=A_0}.
\end{align*}

For the numerator of the main result of Theorem~\ref{th:DensitySamplingMarginalITT}, we thus have
\begin{align*}
\frac{\Exp\big[I(A_0=1)W|Y_K=1\big]}{\Exp\big[I(A_0=0)W|Y_K=1\big]}
&= \frac{\Exp\big[I(A_0=1)WY_K\big]}{\Exp\big[I(A_0=0)WY_K\big]} \\
&= \frac{\Exp\big[WY_K|A_0=1\big]\Pr(A_0=1)}{\Exp\big[WY_K|A_0=0\big]\Pr(A_0=0)},
\end{align*}
where
\begin{align*}
\Exp\big[WY_K|A_0=a\big] &= \Exp\big\{\Exp\big[WY_K|L_0,A_0=a\big]|A_0=a\big\}\\
&= \sum_{l}\frac{\Pr(Y_K=1|L_0=l,A_0=a)\Pr(L_0=l|A_0=a)}{\Pr(A_0=a|L_0=l)}\\
&= \sum_{l}\frac{\Pr(Y_K(a)=1|L_0=l,A_0=a)\Pr(L_0=l|A_0=a)}{\Pr(A_0=a|L_0=l)} \tag{by consistency}\\
&= \sum_{l}\frac{\Pr(Y_K(a)=1|L_0=l)\Pr(L_0=l|A_0=a)}{\Pr(A_0=a|L_0=l)} \tag{by baseline conditional exchangeability}\\
&= \sum_{l}\frac{\Pr(Y_K(a)=1|L_0=l)\Pr(A_0=a|L_0=l)\Pr(L_0=l)}{\Pr(A_0=a|L_0=l)\Pr(A_0=a)}\\
&= \frac{1}{\Pr(A_0=a)}\sum_{l}\Pr(Y_K(a)=1,L_0=l)\\
&= \frac{\Pr(Y_K(a)=1)}{\Pr(A_0=a)},
\end{align*}
so that
\begin{align*}
\frac{\Exp\big[I(A_0=1)W|Y_K=1\big]}{\Exp\big[I(A_0=0)W|Y_K=1\big]}
&= \frac{\Pr(Y_K(1)=1)}{\Pr(Y_K(0)=1)}\\
&= \frac{\sum_{k=0}^{K-1}\Pr(Y_{k+1}(1)=1,Y_k(1)=0)}{\sum_{k=0}^{K-1}\Pr(Y_{k+1}(0)=1,Y_k(0)=0)} \\
&= \frac{\sum_{k=0}^{K-1}\Pr(Y_{k+1}(1)=1|Y_k(1)=0)\Pr(Y_k(1)=0)}{\sum_{k=0}^{K-1}\Pr(Y_{k+1}(0)=1|Y_k(0)=0)\Pr(Y_k(0)=0)} \\
&= \frac{\sum_{k=0}^{K-1}\theta_1\Pr(Y_k(1)=0)}{\sum_{k=0}^{K-1}\theta_0\Pr(Y_k(0)=0)} \tag{by \HomogeneousRates{}}\\
&= \frac{\theta_1}{\theta_0}\frac{\sum_{k=0}^{K-1}\Pr(Y_k(1)=0)}{\sum_{k=0}^{K-1}\Pr(Y_k(0)=0)}
\end{align*}

For the denominator, we have
\begin{align*}
\frac{\Exp\big[I(A_0=1)W\sum_{k=0}^{K-1}S_k\big]}{\Exp\big[I(A_0=0)W\sum_{k=0}^{K-1}S_k\big]}&= \frac{\Exp\big[W\sum_{k=0}^{K-1}S_k|A_0=1\big]\Pr(A_0=1)}{\Exp\big[W\sum_{k=0}^{K-1}S_k|A_0=0\big]\Pr(A_0=0)},
\end{align*}
where \begin{align*}
\textstyle\Exp\big[W\sum_{k=0}^{K-1}S_k|A_0=a\big] &= \textstyle\sum_{k=0}^{K-1}\Exp\big\{\Exp\big[WS_k|L_0,A_0=a\big]|A_0=a\big\}\\ &=\sum_{k=0}^{K-1}\sum_l\frac{\Pr(S_k=1|L_0,A_0=a)\Pr(L_0=l|A_0=a)}{\Pr(A_0=a|L_0=l)}\\
&=\sum_{k=0}^{K-1}\sum_l\frac{\delta\Pr(Y_k=0|L_0=l,A_0=a)\Pr(L_0=l|A_0=a)}{\Pr(A_0=a|L_0=l)} \tag{by \RiskSetSamplingNoMatching{}}\\
&=\sum_{k=0}^{K-1}\sum_l\frac{\delta\Pr(Y_k=0|L_0=l,A_0=a)\Pr(L_0=l)}{\Pr(A_0=a)}\\
&=\sum_{k=0}^{K-1}\sum_l\frac{\delta\Pr(Y_k(a)=0|L_0=l,A_0=a)\Pr(L_0=l)}{\Pr(A_0=a)} \tag{by consistency}\\
&=\sum_{k=0}^{K-1}\sum_l\frac{\delta\Pr(Y_k(a)=0|L_0=l)\Pr(L_0=l)}{\Pr(A_0=a)} \tag{by baseline conditional exchangeability}\\
&=\frac{1}{\Pr(A_0=a)}\sum_{k=0}^{K-1}\sum_l\delta\Pr(Y_k(a)=0,L_0=l)\\
&=\frac{1}{\Pr(A_0=a)}\sum_{k=0}^{K-1}\delta\Pr(Y_k(a)=0),
\end{align*}
so that
\begin{align*}
\frac{\Exp\big[I(A_0=1)W\sum_{k=0}^{K-1}S_k\big]}{\Exp\big[I(A_0=0)W\sum_{k=0}^{K-1}S_k\big]}&= \frac{\sum_{k=0}^{K-1}\delta\Pr(Y_k(1)=0)}{\sum_{k=0}^{K-1}\delta\Pr(Y_k(0)=0)}\\
&= \frac{\sum_{k=0}^{K-1}\Pr(Y_k(1)=0)}{\sum_{k=0}^{K-1}\Pr(Y_k(0)=0)}. %\tag{since $\delta=\delta_{k'}$ for all $k=k'$}.%\\
% &= \frac{1+\sum_{k=1}^{K-1}\prod_{j=0}^{k-1}\Pr(Y_{j+1}(1)=0|Y_{j}(1)=0)}{1+\sum_{k=1}^{K-1}\prod_{j=0}^{k-1}\Pr(Y_{j+1}(0)=0|Y_{j}(0)=0)}\\
% &= \frac{1+\sum_{k=1}^{K-1}\prod_{j=0}^{k-1}(1-\theta_1)}{1+\sum_{k=1}^{K-1}\prod_{j=0}^{k-1}(1-\theta_0)} \tag{by \HomogeneousRates{}}\\
% &= \frac{1+\sum_{k=1}^{K-1}(1-\theta_1)^k}{1+\sum_{k=1}^{K-1}(1-\theta_0)^k}\\
% &= \frac{1+[(1-\theta_1)-(1-\theta_1)^{K-1}]/\theta_1}{1+[(1-\theta_0)-(1-\theta_0)^{K-1}]/\theta_0} \tag{since $(1-r)\sum_{k=l}^uar^k=a(r^l-r^{u+1})$ for any real $a,r$}\\
% &= \frac{\theta_{0}(1-(1-\theta_1)^{K-1})}{\theta_{1}(1-(1-\theta_{0})^{K-1})}.
\end{align*}

It follows that \begin{align*}
\frac{\displaystyle\frac{\Exp\big[I(A_0=1)W|Y_K=1\big]}{\Exp\big[I(A_0=0)W|Y_K=1\big]}}{\displaystyle
\frac{\Exp\big[I(A_0=1)W\sum_{k=0}^{K-1}S_k\big]}{\Exp\big[I(A_0=0)W\sum_{k=0}^{K-1}S_k\big]}
} &= \frac{\Pr(Y_{k+1}(1)=1|Y_{k}(1)=0)}{\Pr(Y_{k+1}(0)=1|Y_{k}(0)=0)}.
\end{align*}
\end{proof}

\begin{remark}[Remark to Theorem~\ref{th:DensitySamplingMarginalITT}]
Condition \RiskSetSamplingNoMatching{} holds if, for some constant $\delta^\ast_k$,
\begin{gather}
\left.\begin{array}{c}
\Pr(S_k=1)=\delta^\ast_k\Pr(Y_{k+1}=1,Y_{k}=0),\\
%S_k\CI(\overline{L},\overline{A},\overline{Y}_K,\overline{S}_{k-1})|Y_k=0,\\
%S_k\CI(\overline{L}_k,\overline{A}_k,\overline{Y}_{k},\overline{S}_{k-1})|Y_k=0,\\
S_k\CI(L_0,A_0,\overline{Y}_{k})|Y_k=0,\\
\Pr(S_k=1|Y_k=1)=0.
\end{array}\right\}\tag{\RiskSetSamplingNoMatching$^\ast$}
\end{gather}
The first requirement of \RiskSetSamplingNoMatching$^\ast$ essentially means that the frequency of incident cases in the $k$th window is proportional to the frequency of controls selected in this window. Under \RiskSetSamplingNoMatching$^\ast$, \RiskSetSamplingNoMatching{} is met with $\delta=\delta^\ast_k\Pr(Y_{k+1}=1|Y_{k}=0)$, because
\begin{align*}
%\Pr(S_k=1|\overline{L},\overline{A},\overline{Y}_K,\overline{S}_{k-1}) &=
%\Pr(S_k=1|\overline{L}_k,\overline{A}_k,\overline{Y}_{k},\overline{S}_{k-1}) &=
\Pr(S_k=1|L_0,A_0,\overline{Y}_{k}) &=
\Pr(S_k=1|Y_k) \\
&=
\Pr(S_k=1|Y_k=0)\times (1-Y_k) \\
&= \Pr(S_k=1|Y_k=0)\times (1-Y_k) \\
&= \frac{\Pr(S_k=1)}{\Pr(Y_k=0)}\times (1-Y_k)\\
&= \frac{\delta^\ast_k\Pr(Y_{k+1}=1,Y_{k}=0)}{\Pr(Y_k=0)}\times (1-Y_k)\\
&= \delta^\ast_k\Pr(Y_{k+1}=1|Y_{k}=0)\times (1-Y_k).
\end{align*}
Therefore, if $\delta^\ast_k$ is $k$-invariant is to state that $\Pr(Y_{k+1}=1|Y_{k}=0)$ is constant for $k=0,...,K-1$.
\end{remark}

\begin{theorem}[\titleDensitySamplingConditionalITT]\label{th:DensitySamplingConditionalITT}
Suppose BCE holds as well as \RiskSetSamplingNoMatching{}, or the weaker version $\Pr(S_k=1|L_0,A_0,Y_{k}) = \Pr(S_k=1|L_0,Y_{k}) = 
\delta_{L_0}\times(1-Y_k),~\delta_{L_0}\in(0,1]$. If  \begin{align*}
\Pr(Y_{k+1}(a)=1|L_0=l,Y_k(a)=0)=\theta_{a} \tag{\HomogeneousConditionalRates{}}
\end{align*}
for $a=0,1$, all $l$ and some constants $\theta_0,\theta_1$, then
\begin{align*}
\frac{\displaystyle\frac{\Exp\big[I(A_0=1)|L_0,Y_K=1\big]}{\Exp\big[I(A_0=0)|L_0,Y_K=1\big]}}{\displaystyle
\frac{\Exp\big[I(A_0=1)\sum_{k=0}^{K-1}S_k|L_0\big]}{\Exp\big[I(A_0=0)\sum_{k=0}^{K-1}S_k|L_0\big]}
} &= \frac{\Pr(Y_{k+1}(1)=1|L_0,Y_{k}(1)=0)}{\Pr(Y_{k+1}(0)=1|L_0,Y_{k}(0)=0)}.
\end{align*}
\end{theorem}
The proof to Theorem~\ref{th:DensitySamplingConditionalITT} is similar to that of Theorem~\ref{th:DensitySamplingMarginalITT} and therefore omitted.

\subsection*{Per-protocol effect}
In this subsection, an individual qualifies as a case if and only if $Y_K=1$ and the subject adheres to the protocol that was assigned at baseline. For any study participant, let $S_k$ denote selection as a control for the period $[t_k,t_{k+1})$ and suppose $S_k$ satisfies
% \begin{align*}
% \Pr(S_k=1|\overline{L}_k,\overline{A}_k,\overline{Y}_k)%,\overline{S}_{k-1}) 
% = \Pr(S_k=1|\overline{A}_{k-1},Y_k) 
% &= \delta\times(1-Y_k)\times I(\forall j<k:A_j= A_0) \tag{\PPRiskSetSamplingNoMatching{}}
% \end{align*}
\begin{align*}
\left.\begin{array}{c}
S_k=1\Rightarrow Y_k=0\text{~~with probability 1, and}\\
\Pr(S_k=1|\overline{L}_k,\overline{A}_k,Y_k=0)%,\overline{S}_{k-1}) 
= \Pr(S_k=1|\overline{A}_{k-1},Y_k=0)\text{~~and}\\
\Pr(S_k=1|\overline{A}_{k-1},A_0=...=A_{k-1},Y_k=0) = \delta,\end{array}\right\} \tag{\PPRiskSetSamplingNoMatching{}}
\end{align*}
for some $\delta\in(0,1]$.

\begin{remark}[Remark to Theorem~\ref{th:DensitySamplingPP}]
Condition \PPRiskSetSamplingNoMatching{} holds if, for some constant $\delta^\ast_k$,
\begin{gather}
\left.\begin{array}{c}
\Pr(S_k=1)=\delta^\ast_k\Pr(Y_{k+1}=1,Y_{k}=0,\forall j<k:A_j= A_0)\text{~~and}\\
S_k\CI(\overline{L}_k,\overline{A}_k,\overline{Y}_k)|(Y_k=0,\forall j<k:A_j=A_0)\text{~~and}\\
S_k=1\Rightarrow (Y_k=0,\forall j<k:A_j=A_0) \text{~~with probability 1}.
\end{array}\right\}\tag{\PPRiskSetSamplingNoMatching$^\ast$}
\end{gather}
The first requirement of \PPRiskSetSamplingNoMatching$^\ast$ essentially means that the frequency of protocol-adherent incident cases in the $k$th window is proportional to the frequency of controls selected in this window. Under \PPRiskSetSamplingNoMatching$^\ast$, \PPRiskSetSamplingNoMatching{} is met with $\delta=\delta^\ast_k\Pr(Y_{k+1}=1|Y_{k}=0,\forall j<k:A_j= A_0)$, because
\begin{align*}
&\Pr(S_k=1|\overline{L}_k,\overline{A}_k,\overline{Y}_k) \\
&\qquad=\Pr(S_k=1|Y_k=0,\forall j<k:A_j= A_0)\times (1-Y_k)\times I(\forall j<k:A_j= A_0) \\
&\qquad= \frac{\Pr(S_k=1)}{\Pr(Y_k=0,\forall j<k:A_j= A_0)}\times (1-Y_k)\times I(\forall j<k:A_j= A_0) \\
&\qquad= \frac{\delta^\ast_k\Pr(Y_{k+1}=1,Y_{k}=0,\forall j<k:A_j= A_0)}{\Pr(Y_k=0,\forall j<k:A_j= A_0)}\times (1-Y_k)\times I(\forall j<k:A_j= A_0) \\
&\qquad= \delta^\ast_k\Pr(Y_{k+1}=1|Y_{k}=0,\forall j<k:A_j= A_0)\times (1-Y_k)\times I(\forall j<k:A_j= A_0).
\end{align*}

Similarly, condition \PPRiskSetSamplingNoMatching{} holds if, for some constant $\delta^{\ast\ast}_k$,
\begin{gather}
\left.\begin{array}{c}
\Pr(S_k=1)=\delta^{\ast\ast}_k\Pr(Y_{k+1}=1,Y_{k}=0)\text{~~and}\\
S_k\CI(\overline{L}_k,\overline{A}_k,\overline{Y}_k)|(Y_k=0)\text{~~and}\\
S_k=1\Rightarrow Y_k=0 \text{~~with probability 1},
\end{array}\right\}\tag{\PPRiskSetSamplingNoMatching$^{\ast\ast}$}
\end{gather}
in which case, $\delta=\delta^{\ast\ast}_k\Pr(Y_{k+1}=1|Y_{k}=0)$, because
\begin{align*}
&\Pr(S_k=1|\overline{L}_k,\overline{A}_k,\overline{Y}_k) \\
&\qquad=\Pr(S_k=1|Y_k=0)\times (1-Y_k) \\
&\qquad= \frac{\Pr(S_k=1)}{\Pr(Y_k=0)}\times (1-Y_k)\\
&\qquad= \frac{\delta^{\ast\ast}_k\Pr(Y_{k+1}=1,Y_{k}=0)}{\Pr(Y_k=0)}\times (1-Y_k)\\
&\qquad= \delta^{\ast\ast}_k\Pr(Y_{k+1}=1|Y_{k}=0)\times (1-Y_k).
\end{align*}
\end{remark}

\begin{theorem}[\titleDensitySamplingPP]\label{th:DensitySamplingPP}
Suppose SCE and \PPRiskSetSamplingNoMatching{} hold. 
If \begin{align*}
\Pr(Y_{k+1}(\overline{a})=1|Y_k(\overline{a})=0)=\theta_a \tag{\PPHomogeneousRates{}}
\end{align*}
for $a=0,1$ and some constants $\theta_0,\theta_1$,
then
\begin{align*}
% &\frac{\displaystyle\frac{\Exp\big[\sum_{k=0}^{K-1}I(A_k=1)W_kI(Y_{k+1}=1,Y_k=0,\forall j\le k:A_j=A_0)|Y_K=1\big]}{\Exp\big[\sum_{k=0}^{K-1}I(A_k=0)W_kI(Y_{k+1}=1,Y_k=0,\forall j\le k:A_j=A_0)|Y_K=1\big]}}{\displaystyle
% \frac{\Exp\big[\sum_{k=0}^{K-1}I(A_k=1)W_kS_k\big]}{\Exp\big[\sum_{k=0}^{K-1}I(A_k=0)W_kS_k\big]}
% } \\
&\frac{\displaystyle\frac{\Exp\big[\sum_{k=0}^{K-1}I(A_k=1)W_kI(Y_{k+1}=1,Y_k=0)|Y_K=1,(\forall j:Y_j=0\Rightarrow A_j=A_0)\big]}{\Exp\big[\sum_{k=0}^{K-1}I(A_k=0)W_kI(Y_{k+1}=1,Y_k=0)|Y_K=1,(\forall j:Y_j=0\Rightarrow A_j=A_0)\big]}}{\displaystyle \frac{\Exp\big[I(A_0=1)\sum_{k=0}^{K-1}W_kS_k|\forall j:Y_j=0\Rightarrow A_j=A_0\big]}{\Exp\big[I(A_0=0)\sum_{k=0}^{K-1}W_kS_k|\forall j:Y_j=0\Rightarrow A_j=A_0\big]}} 
%{\displaystyle \frac{\Exp\big[\sum_{k=0}^{K-1}I(A_k=1)W_kS_k|\forall j:Y_j=0\Rightarrow A_j=A_0\big]}{\Exp\big[\sum_{k=0}^{K-1}I(A_k=0)W_kS_k|\forall j:Y_j=0\Rightarrow A_j=A_0\big]}}
\\
&\qquad= \frac{\Pr(Y_{k+1}(\overline{1})=1|Y_{k}(\overline{1})=0)}{\Pr(Y_{k+1}(\overline{0})=1|Y_{k}(\overline{0})=0)},
\end{align*}
where 
\begin{align*}
W_k=
\prod_{j=0}^k\frac{1}{\Pr(A_j=a_j|\overline{L}_j,\overline{A}_{j-1},Y_j=0,S_j=1)}\bigg|_{a_j=A_j}.
\end{align*}
\end{theorem}
\begin{proof}
First, observe that $\Pr(A_k=a'|\overline{L}_k,(\forall j<k:A_{j}=a),Y_k=0,S_k=1)=\Pr(A_k=a'|\overline{L}_k,(\forall j<k:A_{j}=a),Y_k=0)$ for $a',a=0,1$, because
\begin{align*}
&\Pr(A_k=a'|\overline{L}_k,(\forall j<k:A_{j}=a),Y_k=0,S_k=1)\\
&\qquad=\frac{\Pr(S_k=1|\overline{L}_k,(\forall j<k:A_{j}=a),A_k=a',Y_k=0)\Pr(A_k=a'|\overline{L}_k,(\forall j<k:A_{j}=a),Y_k=0)}{\Pr(S_k=1|\overline{L}_k,(\forall j<k:A_{j}=a),Y_k=0)} \\
&\qquad=\frac{\delta}{\delta}\Pr(A_k=a'|\overline{L}_k,(\forall j<k:A_{j}=a),Y_k=0). \tag{by \PPRiskSetSamplingNoMatching{}}
\end{align*}
Hence, if $\forall j<k:A_j=A_0$, then 
\begin{align*}
W_k=
\prod_{j=0}^k\frac{1}{\Pr(A_j=a_j|\overline{L}_j,\overline{A}_{j-1},Y_j=0)}\bigg|_{a_j=A_j}.
\end{align*}

For the numerator of the main result of Theorem~\ref{th:DensitySamplingPP}, we thus have
\begin{align*}
&
{\displaystyle\frac{\Exp\big[\sum_{k=0}^{K-1}I(A_k=1)W_kI(Y_{k+1}=1,Y_k=0)|Y_K=1,(\forall j:Y_j=0\Rightarrow A_j=A_0)\big]}{\Exp\big[\sum_{k=0}^{K-1}I(A_k=0)W_kI(Y_{k+1}=1,Y_k=0)|Y_K=1,(\forall j:Y_j=0\Rightarrow A_j=A_0)\big]}}\\
%\frac{\Exp\big[\sum_{k=0}^{K-1}I(A_k=a)W_kI(Y_{k+1}=1,Y_k=0,\forall j\le k:A_j=A_0)|Y_K=1\big]}{\Exp\big[\sum_{k=0}^{K-1}I(A_k=a')W_kI(Y_{k+1}=1,Y_k=0,\forall j\le k:A_j=A_0)|Y_K=1\big]}\\
%%%%
&\qquad\frac{\Exp\big[\sum_{k=0}^{K-1}I(A_k=a)W_kI(Y_{k+1}=1,Y_k=0,\forall j\le k:A_j=A_0)\big]}{\Exp\big[\sum_{k=0}^{K-1}I(A_k=a')W_kI(Y_{k+1}=1,Y_k=0,\forall j\le k:A_j=A_0)\big]}\\
%%%%
&\qquad=\frac{\sum_{k=0}^{K-1}\Exp\big[W_kY_{k+1}(1-Y_k)I(\forall j\le k:A_j=a)\big]}{\sum_{k=0}^{K-1}\Exp\big[W_kY_{k+1}(1-Y_k)I(\forall j\le k:A_j=a')\big]}\\
%%%%
&\qquad=\frac{\displaystyle\sum_{k=0}^{K-1}\sum_{\overline{l}_k}\frac{\Pr(Y_{k+1}=1,Y_k=0,\forall j\le k: A_j=a,\overline{L}_k=\overline{l}_k)}{\prod_{j=0}^k\Pr(A_j=a|Y_j=0,\overline{L}_k=\overline{l}_k,\forall i<j: A_i=a)}}
{\displaystyle\sum_{k=0}^{K-1}\sum_{\overline{l}_k}\frac{\Pr(Y_{k+1}=1,Y_k=0,\forall j\le k: A_j=a',\overline{L}_k=\overline{l}_k)}{\prod_{j=0}^k\Pr(A_j=a'|Y_j=0,\overline{L}_k=\overline{l}_k,\forall i<j: A_i=a')}},
\end{align*}
where
\begin{align*}
&\sum_{\overline{l}_k}\frac{\Pr(Y_{k+1}=1,Y_k=0,\forall j\le k: A_j=a,\overline{L}_k=\overline{l}_k)}{\prod_{j=0}^k\Pr(A_j=a|Y_j=0,\overline{L}_k=\overline{l}_k,\forall i<j: A_i=a)}\\
&\qquad=\sum_{\overline{l}_k}\Pr(Y_{k+1}=1|Y_k=0,\overline{L}_k=\overline{l}_k,\forall j\le k: A_j=a)\\
&\qquad\qquad\times \Pr(L_k=l_k|Y_k=0,\overline{L}_{k-1}=\overline{l}_{k-1},\forall j<k: A_j=a)\\
&\qquad\qquad\times\prod_{j=0}^{k-1}\Pr(Y_{j+1}=1|Y_j=0,\overline{L}_j=\overline{l}_j,\forall i\le j: A_i=a)\\
&\qquad\qquad\times \Pr(L_j=l_j|Y_j=0,\overline{L}_{j-1}=\overline{l}_{j-1},\forall i<j: A_i=a)\\
%%%%
&\qquad=\sum_{\overline{l}_k}\Pr(Y_{k+1}(\overline{a})=1|Y_k(\overline{a})=0,\overline{L}_k=\overline{l}_k,\forall j\le k: A_j=a)\\
&\qquad\qquad\times \Pr(L_k=l_k|Y_k(\overline{a})=0,\overline{L}_{k-1}=\overline{l}_{k-1},\forall j<k: A_j=a)\\
&\qquad\qquad\times\prod_{j=0}^{k-1}\Pr(Y_{j+1}(\overline{a})=1|Y_j(\overline{a})=0,\overline{L}_j=\overline{l}_j,\forall i\le j: A_i=a)\\
&\qquad\qquad\times \Pr(L_j=l_j|Y_j(\overline{a})=0,\overline{L}_{j-1}=\overline{l}_{j-1},\forall i<j: A_i=a) \tag{by consistency}\\
%%%%
&\qquad=\sum_{\overline{l}_k}\Pr(Y_{k+1}(\overline{a})=1|Y_k(\overline{a})=0,\overline{L}_k=\overline{l}_k,\forall j< k: A_j=a)\\
&\qquad\qquad\times \Pr(L_k=l_k|Y_k(\overline{a})=0,\overline{L}_{k-1}=\overline{l}_{k-1},\forall j<k: A_j=a)\\
&\qquad\qquad\times\prod_{j=0}^{k-1}\Pr(Y_{j+1}(\overline{a})=1|Y_j(\overline{a})=0,\overline{L}_j=\overline{l}_j,\forall i< j: A_i=a)\\
&\qquad\qquad\times \Pr(L_j=l_j|Y_j(\overline{a})=0,\overline{L}_{j-1}=\overline{l}_{j-1},\forall i<j: A_i=a) \tag{by sequential conditional exchangeability}\\
%%%%
&\qquad=\sum_{\overline{l}_{k-1}}\Pr(Y_{k+1}(\overline{a})=1|Y_{k}(\overline{a})=0,\overline{L}_{k-1}=\overline{l}_{k-1},\forall j< k: A_j=a)\\
&\qquad\qquad\times\prod_{j=0}^{k-1}\Pr(Y_{j+1}(\overline{a})=1|Y_j(\overline{a})=0,\overline{L}_j=\overline{l}_j,\forall i< j: A_i=a)\\
&\qquad\qquad\times \Pr(L_j=l_j|Y_j(\overline{a})=0,\overline{L}_{j-1}=\overline{l}_{j-1},\forall i<j: A_i=a)\\
%%%%
&\qquad=\sum_{\overline{l}_{k-1}}\Pr(Y_{k+1}(\overline{a})=1,Y_k(\overline{a})=0|Y_{k-1}(\overline{a})=0,\overline{L}_{k-1}=\overline{l}_{k-1},\forall j< k: A_j=a)\\
&\qquad\qquad\times\Pr(L_{k-1}=l_{k-1}|Y_{k-1}(\overline{a})=0,\overline{L}_{k-2}=\overline{l}_{k-2},\forall j<k-1: A_j=a)\\
&\qquad\qquad\times\prod_{j=0}^{k-2}\Pr(Y_{j+1}(\overline{a})=1|Y_j(\overline{a})=0,\overline{L}_j=\overline{l}_j,\forall i< j: A_i=a)\\
&\qquad\qquad\times \Pr(L_j=l_j|Y_j(\overline{a})=0,\overline{L}_{j-1}=\overline{l}_{j-1},\forall i<j: A_i=a)\\
%%%%
&\qquad~~\vdots\tag{by repeating previous three steps, under sequential conditional exchangeability}\\
&\qquad=\Pr(Y_{k+1}(\overline{a})=1,Y_k(\overline{a})=0)
\end{align*}
and, similarly,
\begin{align*}
\sum_{\overline{l}_k}\frac{\Pr(Y_{k+1}=1,Y_k=0,\forall j\le k: A_j=a',\overline{L}_k=\overline{l}_k)}{\prod_{j=0}^k\Pr(A_j=a'|Y_j=0,\overline{L}_k=\overline{l}_k,\forall i<j: A_i=a')}=\Pr(Y_{k+1}(\overline{a}')=1,Y_k(\overline{a}')=0).
\end{align*}
Hence,
\begin{align*}
&\frac{\Exp\big[\sum_{k=0}^{K-1}I(A_k=a)W_kI(Y_{k+1}=1,Y_k=0,\forall j\le k:A_j=A_0)\big]}{\Exp\big[\sum_{k=0}^{K-1}I(A_k=a')W_kI(Y_{k+1}=1,Y_k=0,\forall j\le k:A_j=A_0)\big]}\\
%%%%
&\qquad=\frac{\sum_{k=0}^{K-1}\Pr(Y_{k+1}(\overline{a})=1,Y_k(\overline{a})=0)}
{\sum_{k=0}^{K-1}\Pr(Y_{k+1}(\overline{a}')=1,Y_k(\overline{a}')=0)}\\
%%%%
&\qquad=\frac{\sum_{k=0}^{K-1}\Pr(Y_{k+1}(\overline{a})=1|Y_k(\overline{a})=0)\prod_{j=1}^k\Pr(Y_{j}(\overline{a})=0|Y_{j-1}(\overline{a})=0)}
{\sum_{k=0}^{K-1}\Pr(Y_{k+1}(\overline{a}')=1|Y_k(\overline{a}')=0)\prod_{j=1}^k\Pr(Y_{j}(\overline{a}')=0|Y_{j-1}(\overline{a}')=0)}\\
%%%%
&\qquad=\frac{\sum_{k=0}^{K-1}\theta_a(1-\theta_a)^k}
{\sum_{k=0}^{K-1}\theta_{a'}(1-\theta_{a'})^k} \tag{\PPHomogeneousRates{}}\\
%%%%
&\qquad=\frac{1-(1-\theta_a)^{K}}{1-(1-\theta_{a'})^{K}} \tag{since $(1-r)\sum_{k=l}^uar^k=a(r^l-r^{u+1})$ for any real $a,r$}
\end{align*}

For the denominator, we have
\begin{align*}
&\frac{\Exp\big[I(A_0=a)\sum_{k=0}^{K-1}W_kS_k|\forall j:Y_j=0\Rightarrow A_j=A_0\big]}{\Exp\big[I(A_0=a')\sum_{k=0}^{K-1}W_kS_k|\forall j:Y_j=0\Rightarrow A_j=A_0\big]}\\
&\qquad=\frac{\Exp\big[\sum_{k=0}^{K-1}I(A_k=a)W_kS_k|\forall j:Y_j=0\Rightarrow A_j=A_0\big]}{\Exp\big[\sum_{k=0}^{K-1}I(A_k=a')W_kS_k|\forall j:Y_j=0\Rightarrow A_j=A_0\big]}\\
%%%%
&\qquad=\frac{\sum_{k=0}^{K-1}\Exp\big[I(A_k=a)W_kS_k|\forall j:Y_j=0\Rightarrow A_j=A_0\big]}{\sum_{k=0}^{K-1}\Exp\big[I(A_k=a')W_kS_k|\forall j:Y_j=0\Rightarrow A_j=A_0\big]}\\
%%%%
&\qquad=\frac{\sum_{k=0}^{K-1}\Exp\big[I(A_k=a)W_kS_k|Y_k=0,\forall j\le k:A_j=A_0\big]\Pr(Y_k=0|\forall j:Y_j=0\Rightarrow A_j=A_0)}{\sum_{k=0}^{K-1}\Exp\big[I(A_k=a')W_kS_k|Y_k=0,\forall j\le k:A_j=A_0\big]\Pr(Y_k=0|\forall j:Y_j=0\Rightarrow A_j=A_0)} \tag{by \PPRiskSetSamplingNoMatching{}}\\
%%%%
&\qquad=\frac{\sum_{k=0}^{K-1}\Exp\big[I(A_k=a)W_kS_k|Y_k=0,\forall j\le k:A_j=A_0\big]\Pr(Y_k=0,\forall j\le k: A_j=A_0)}{\sum_{k=0}^{K-1}\Exp\big[I(A_k=a')W_kS_k|Y_k=0,\forall j\le k:A_j=A_0\big]\Pr(Y_k=0,\forall j\le k: A_j=A_0)}\\
%%%%
&\qquad=\frac{\sum_{k=0}^{K-1}\Exp\big[W_kS_k|Y_k=0,\forall j\le k:A_j=a\big]\Pr(Y_k=0,\forall j\le k: A_j=a)}{\sum_{k=0}^{K-1}\Exp\big[W_kS_k|Y_k=0,\forall j\le k:A_j=a'\big]\Pr(Y_k=0,\forall j\le k: A_j=a')}\\
%%%%
&\qquad=\frac{\displaystyle\sum_{k=0}^{K-1}\sum_{\overline{l}_k}\frac{\Exp\big[S_k|Y_k=0,\overline{L}_k=\overline{l}_k,\forall j\le k:A_j=a\big]\Pr(Y_k=0,\overline{L}_k=\overline{l}_k,\forall j\le k:A_j=a)}{\prod_{j=0}^k\Pr(A_j=a|Y_j=0,\overline{L}_j=\overline{l}_j,\forall i< j:A_i=a)}}{\displaystyle\sum_{k=0}^{K-1}\sum_{\overline{l}_k}\frac{\Exp\big[S_k|Y_k=0,\overline{L}_k=\overline{l}_k,\forall j\le k:A_j=a'\big]\Pr(Y_k=0,\overline{L}_k=\overline{l}_k,\forall j\le k:A_j=a')}{\prod_{j=0}^k\Pr(A_j=a'|Y_j=0,\overline{L}_j=\overline{l}_j,\forall i< j:A_i=a')}}\\
%%%%
&\qquad=\frac{\displaystyle\sum_{k=0}^{K-1}\sum_{\overline{l}_k}\delta\frac{\Pr(Y_k=0,\overline{L}_k=\overline{l}_k,\forall j\le k:A_j=a)}{\prod_{j=0}^k\Pr(A_j=a|Y_j=0,\overline{L}_j=\overline{l}_j,\forall i< j:A_i=a)}}{\displaystyle\sum_{k=0}^{K-1}\sum_{\overline{l}_k}\delta\frac{\Pr(Y_k=0,\overline{L}_k=\overline{l}_k,\forall j\le k:A_j=a')}{\prod_{j=0}^k\Pr(A_j=a'|Y_j=0,\overline{L}_j=\overline{l}_j,\forall i< j:A_i=a')}}\tag{by \PPRiskSetSamplingNoMatching{}}\\
%%%%
&\qquad=\frac{\sum_{k=0}^{K-1}\sum_{\overline{l}_k}\delta\prod_{j=0}^k\Pr(Y_j=0,L_j=l_j|Y_{j-1}=0,\overline{L}_{j-1}=\overline{l}_{j-1},\forall i<j:A_i=a)}{\sum_{k=0}^{K-1}\sum_{\overline{l}_k}\delta\prod_{j=0}^k\Pr(Y_j=0,L_j=l_j|Y_{j-1}=0,\overline{L}_{j-1}=\overline{l}_{j-1},\forall i<j:A_i=a')}\\
%%%%
&\qquad=\frac{\displaystyle\sum_{k=0}^{K-1}\sum_{\overline{l}_k}\delta\prod_{j=0}^k\begin{array}{l}\Pr(L_j=l_j|Y_j=0,\overline{L}_{j-1}=\overline{l}_{j-1},\forall i<j:A_i=a)\times{}\\~~~\Pr(Y_j=0|Y_{j-1}=0,\overline{L}_{j-1}=\overline{l}_{j-1},\forall i<j:A_i=a)\end{array}}{\displaystyle\sum_{k=0}^{K}-1\sum_{\overline{l}_k}\delta\prod_{j=0}^k\begin{array}{l}\Pr(L_j=l_j|Y_j=0,\overline{L}_{j-1}=\overline{l}_{j-1},\forall i<j:A_i=a')\times{}\\~~~\Pr(Y_j=0|Y_{j-1}=0,\overline{L}_{j-1}=\overline{l}_{j-1},\forall i<j:A_i=a')\end{array}}\\
%%%%
&\qquad=\frac{\displaystyle\sum_{k=0}^{K-1}\sum_{\overline{l}_k}\delta\prod_{j=0}^k\begin{array}{l}\Pr(L_j=l_j|Y_j(\overline{a})=0,\overline{L}_{j-1}=\overline{l}_{j-1},\forall i<j:A_i=a)\times{}\\~~~\Pr(Y_j(\overline{a})=0|Y_{j-1}(\overline{a})=0,\overline{L}_{j-1}=\overline{l}_{j-1},\forall i<j:A_i=a)\end{array}}{\displaystyle\sum_{k=0}^{K-1}\sum_{\overline{l}_k}\delta\prod_{j=0}^k\begin{array}{l}\Pr(L_j=l_j|Y_j(\overline{a}')=0,\overline{L}_{j-1}=\overline{l}_{j-1},\forall i<j:A_i=a')\times{}\\~~~\Pr(Y_j(\overline{a}')=0|Y_{j-1}(\overline{a}')=0,\overline{L}_{j-1}=\overline{l}_{j-1},\forall i<j:A_i=a')\end{array}} \tag{by consistency}\\
%%%%
&\qquad=\frac{\displaystyle\sum_{k=0}^{K-1}\sum_{\overline{l}_{k-1}}\delta\prod_{j=0}^k\begin{array}{l}\Pr(Y_j(\overline{a})=0|Y_{j-1}(\overline{a})=0,\overline{L}_{j-1}=\overline{l}_{j-1},\forall i<j:A_i=a)\times{}\\~~~\Pr(L_{j-1}=l_{j-1}|Y_{j-1}(\overline{a})=0,\overline{L}_{j-2}=\overline{l}_{j-2},\forall i<j-1:A_i=a)\end{array}}{\displaystyle\sum_{k=0}^{K-1}\sum_{\overline{l}_{k-1}}\delta\prod_{j=0}^k\begin{array}{l}\Pr(Y_j(\overline{a}')=0|Y_{j-1}(\overline{a}')=0,\overline{L}_{j-1}=\overline{l}_{j-1},\forall i<j:A_i=a')\times{}\\~~~\Pr(L_{j-1}=l_{j-1}|Y_{j-1}(\overline{a}')=0,\overline{L}_{j-2}=\overline{l}_{j-2},\forall i<j-1:A_i=a')\end{array}} \\
%%%%
&\qquad=\frac{\displaystyle\sum_{k=0}^{K-1}\sum_{\overline{l}_{k-1}}\delta\prod_{j=0}^k\begin{array}{l}\Pr(Y_j(\overline{a})=0|Y_{j-1}(\overline{a})=0,\overline{L}_{j-1}=\overline{l}_{j-1},\forall i<j-1:A_i=a)\times{}\\~~~\Pr(L_{j-1}=l_{j-1}|Y_{j-1}(\overline{a})=0,\overline{L}_{j-2}=\overline{l}_{j-2},\forall i<j-1:A_i=a)\end{array}}{\displaystyle\sum_{k=0}^{K-1}\sum_{\overline{l}_{k-1}}\delta\prod_{j=0}^k\begin{array}{l}\Pr(Y_j(\overline{a}')=0|Y_{j-1}(\overline{a}')=0,\overline{L}_{j-1}=\overline{l}_{j-1},\forall i<j-1:A_i=a')\times{}\\~~~\Pr(L_{j-1}=l_{j-1}|Y_{j-1}(\overline{a}')=0,\overline{L}_{j-2}=\overline{l}_{j-2},\forall i<j-1:A_i=a')\end{array}} \tag{by sequential conditional exchangeability}\\
&\qquad=\frac{\sum_{k=0}^{K-1}\sum_{\overline{l}_{k-1}}\delta\prod_{j=0}^k\Pr(Y_j(\overline{a})=0,L_{j-1}=l_{j-1}|Y_{j-1}(\overline{a})=0,\overline{L}_{j-2}=\overline{l}_{j-2},\forall i<j-1:A_i=a)}{\sum_{k=0}^{K-1}\sum_{\overline{l}_{k-1}}\delta\prod_{j=0}^k\Pr(Y_j(\overline{a}')=0,L_{j-1}=l_{j-1}|Y_{j-1}(\overline{a}')=0,\overline{L}_{j-2}=\overline{l}_{j-2},\forall i<j-1:A_i=a')}\\
%%%%
&\qquad=\frac{\displaystyle\sum_{k=0}^{K-1}\sum_{\overline{l}_{k-2}}\delta\begin{array}{l}\Pr(Y_k(\overline{a})=0|Y_{k-1}(\overline{a})=0,\overline{L}_{k-2}=\overline{l}_{k-2},\forall i<k-1:A_i=a)\times{}\\~~~\prod_{j=0}^{k-1}\Pr(Y_j(\overline{a})=0,L_{j-1}=l_{j-1}|Y_{j-1}(\overline{a})=0,\overline{L}_{j-2}=\overline{l}_{j-2},\forall i<j-1:A_i=a)\end{array}}{\displaystyle\sum_{k=0}^{K-1}\sum_{\overline{l}_{k-2}}\delta\begin{array}{l}\Pr(Y_k(\overline{a}')=0|Y_{k-1}(\overline{a}')=0,\overline{L}_{k-2}=\overline{l}_{k-2},\forall i<k-1:A_i=a')\times{}\\~~~\prod_{j=0}^{k-1}\Pr(Y_j(\overline{a}')=0,L_{j-1}=l_{j-1}|Y_{j-1}(\overline{a}')=0,\overline{L}_{j-2}=\overline{l}_{j-2},\forall i<j-1:A_i=a')\end{array}}\\
&\qquad~~\vdots \tag{by sequential conditional exchangeability}\\
%%%%
&\qquad=\frac{\sum_{k=0}^{K-1}\delta\Pr(Y_k(\overline{a})=0)}{\sum_{k=0}^{K-1}\delta\Pr(Y_k(\overline{a}')=0)}\\
%%%%
&\qquad=\frac{\sum_{k=0}^{K-1}\Pr(Y_k(\overline{a})=0)}{\sum_{k=0}^{K-1}\Pr(Y_k(\overline{a}')=0)}\\% \tag{since $\delta=\delta_{k'}$ for all $k,k'$}\\
%%%%
&\qquad=\frac{1+\sum_{k=1}^{K-1}\prod_{j=1}^k\Pr(Y_j(\overline{a})=0|Y_{j-1}(\overline{a})=0)}{1+\sum_{k=1}^{K-1}\prod_{j=1}^k\Pr(Y_j(\overline{a}')=0|Y_{j-1}(\overline{a}')=0)}\\
%%%%
&\qquad=\frac{1+\sum_{k=1}^{K-1}(1-\theta_a)^k}{1+\sum_{k=1}^K(1-\theta_{a'})^k}\tag{by \PPHomogeneousRates{}}\\
%%%%
&\qquad=\frac{1+[1-\theta_a-(1-\theta_a)^{K-1}]/\theta_a}{1+[1-\theta_{a'}-(1-\theta_{a'})^{K-1}]/\theta_{a'}} \tag{since $(1-r)\sum_{k=l}^uar^k=a(r^l-r^{u+1})$ for any real $a,r$}\\
&\qquad=\frac{\theta_{a'}(1-(1-\theta_a)^{K-1})}{\theta_{a}(1-(1-\theta_{a'})^{K-1})}.
\end{align*}

Hence,
\begin{align*}
&\displaystyle\frac{\displaystyle\frac{\Exp\big[\sum_{k=0}^{K-1}I(A_k=a)W_kI(Y_{k+1}=1,Y_k=0,\forall j\le k:A_j=A_0)|Y_K=1\big]}{\Exp\big[\sum_{k=0}^{K-1}I(A_k=1-a)W_kI(Y_{k+1}=1,Y_k=0,\forall j\le k:A_j=A_0)|Y_K=1\big]}}{\displaystyle
\frac{\Exp\big[\sum_{k=0}^{K-1}I(A_k=a)W_kS_k\big]}{\Exp\big[\sum_{k=0}^{K-1}I(A_k=1-a)W_kS_k\big]}}\\
&\qquad=\frac{1-(1-\theta_a)^{K-1}}{1-(1-\theta_{a'})^{K-1}}\times{\frac{\theta_{a}(1-(1-\theta_{a'})^{K-1})}{\theta_{a'}(1-(1-\theta_a)^{K-1})}}\\
&\qquad=\theta_a/\theta_{a'},
\end{align*}
which completes the proof.
\end{proof}

\end{myAppendix}

\begin{myAppendix}[: Identification results for exact 1:$M$ matching strategies]\label{app:MatchingExact}

\subsection*{Intention-to-treat effect}
In this subsection, cases are defined by $Y_K=1$ and have baseline exposure $A_0$. All cases are assigned a (possibly variable) number $M\ge 0$ of control exposures $A'_i$, $i=1,...,M$,  subject to
\begin{align*}
\left.\begin{array}{c}
\Pr(M>0|Y_K=1)>0\text{~and}\\
M\CI A_0|(L_0,Y_K=1)\text{~and}\\
%A_0,A'_1,...,A'_M~{\text{mutually independent given}}~(L_0,Y_K,M) \text{~and}\\
\forall{l,a,a'}: \Pr(A'_i=a'|L_0=l,A_0=a,Y_K=1,M,M>0)=\Pr(A_0=a'|L_0=l),
\end{array}\right\}\tag{\BaseSamplingMatching{}}
\end{align*}
or
\begin{align*}
\left.\begin{array}{c}
\Pr(M>0|Y_K=1)>0\text{~and}\\
M\CI A_0|(L_0,Y_K=1)\text{~and}\\
%A_0,A'_1,...,A'_M~{\text{mutually independent given}}~(L_0,Y_K,M) \text{~and}\\
\forall{l,a,a'}: \Pr(A'_i=a'|L_0=l,A_0=a,Y_K=1,M,M>0)=\Pr(A_0=a'|L_0=l,Y_K=0),
\end{array}\right\}\tag{\SurvivorSamplingMatching{}}
\end{align*}
or
\begin{align*}
\left.\begin{array}{c}
\Pr(M>0|Y_K=1)>0\text{~and}\\
M\CI A_0|(L_0,Y_K=1,J)\text{~and}\\
%A_0,A'_1,...,A'_M~{\text{mutually independent given}}~(L_0,\overline{Y}_K,M) \text{~and}\\
\forall{l,a,a'}: \Pr(A'_i=a'|L_0=l,A_0=a,\overline{Y}_K,J=j,M,M>0)=\Pr(A_0=a'|L_0=l,Y_{j}=0),~~\text{where}\\\hspace{5cm} J=\max\{k=0,1,...,K:Y_k=0\}. 
\end{array}\right\}\tag{\RiskSetSamplingMatching{}}
\end{align*}
That is, cases are matched with subjects that have the same baseline covariate level and who are alive at baseline (\BaseSamplingMatching{}), at the end of study (\SurvivorSamplingMatching{}), or whenever the case is alive (\RiskSetSamplingMatching{}).

For simplicity, it is assumed below that the variables are discrete. The results can however be extended to more general distributions.

\begin{theorem}[\titleCaseCohortMatchingITT]\label{th:CaseCohortMatchingITT}
If \BaseSamplingMatching{} and BCE hold and \begin{align*}
\frac{\Pr(Y_K=1|L_0=l,A_0=1)}{\Pr(Y_K=1|L_0=l,A_0=0)} = \theta \tag{\HomogeneousRiskRatios{}}
\end{align*}
for all $l$ and some constant $\theta$, then \begin{align*}
\frac{\Exp\big[\sum_{i=1}^M I(A'_i=0,A_0=1)\big|Y_K=1,M>0\big]}{\Exp\big[\sum_{i=1}^M I(A'_i=1,A_0=0)\big|Y_K=1,M>0\big]} = \frac{\Pr(Y_K(1)=1)}{\Pr(Y_K(0)=1)}.
\end{align*}
\end{theorem}
\begin{proof}
We have
\begin{align*}
\frac{\Exp\big[\sum_{i=1}^M I(A'_i=0,A_0=1)\big|Y_K=1,M>0\big]}{\Exp\big[\sum_{i=1}^M I(A'_i=1,A_0=0)\big|Y_K=1,M>0\big]} &= \frac{\Exp\big[\sum_{i=1}^M I(A'_i=0)\big|A_0=1,Y_K=1,M>0\big]}{\Exp\big[\sum_{i=1}^M I(A'_i=1)\big|A_0=0,Y_K=1,M>0\big]}\\&\qquad{}\times\mathrm{Odds}(A_0=1|Y_K=1,M>0),
\end{align*}
where
\begin{align*}
&\frac{\Exp\big[\sum_{i=1}^M I(A'_i=0)\big|A_0=1,Y_K=1,M>0\big]}{\Exp\big[\sum_{i=1}^M I(A'_i=1)\big|A_0=0,Y_K=1,M>0\big]}\\
&\qquad=\frac{\sum_{m>0}\Exp\big[\sum_{i=1}^m I(A'_i=0)\big|A_0=1,Y_K=1,M=m\big]\Pr(M=m|A_0=1,Y_K=1,M>0)}{\sum_{m>0}\Exp\big[\sum_{i=1}^m I(A'_i=1)\big|A_0=0,Y_K=1,M=m\big]\Pr(M=m|A_0=0,Y_K=1,M>0)}\\
&\qquad=\frac{\sum_{m>0}\sum_{i=1}^m\sum_{l}\Pr(A'_i=0|L_0=l,A_0=1,Y_K=1,M=m)\Pr(M=m,L_0=l|A_0=1,Y_K=1,M>0)}{\sum_{m>0}\sum_{i=1}^m\sum_{l}\Pr(A'_i=1|L_0=l,A_0=0,Y_K=1,M=m)\Pr(M=m,L_0=l|A_0=0,Y_K=1,M>0)} \\
&\qquad=\frac{\sum_{m>0}\sum_{i=1}^m\sum_{l}\Pr(A_0=0|L_0=l)\Pr(M=m,L_0=l|A_0=1,Y_K=1,M>0)}{\sum_{m>0}\sum_{i=1}^m\sum_{l}\Pr(A_0=1|L_0=l)\Pr(M=m,L_0=l|A_0=0,Y_K=1,M>0)}\tag{by \BaseSamplingMatching{}}\\
&\qquad=\frac{\sum_{m>0}\sum_{i=1}^m\sum_{l}\Pr(A_0=0|L_0=l)\Pr(M=m,L_0=l,A_0=1|Y_K=1)}{\sum_{m>0}\sum_{i=1}^m\sum_{l}\Pr(A_0=1|L_0=l)\Pr(M=m,L_0=l,A_0=0|Y_K=1)}\frac{1}{\mathrm{Odds}(A_0=1|Y_K=1,M>0)}\\
&\qquad=\frac{\sum_{m>0}\sum_{i=1}^m\sum_{l}q(l,m)\Pr(Y_K=1|L_0=l,A_0=1)}{\sum_{m>0}\sum_{i=1}^m\sum_{l}q(l,m)\Pr(Y_K=1|L_0=l,A_0=0)} \\&\qquad\qquad\times\frac{1}{\mathrm{Odds}(A_0=1|Y_K=1,M>0)} \tag{under \BaseSamplingMatching{} and definition of $q(l,m)$ (see below)}\\
&\qquad=\frac{\sum_{m>0}\sum_{i=1}^m\sum_{l}q(l,m)\theta\Pr(Y_K=1|L_0=l,A_0=0)}{\sum_{m>0}\sum_{i=1}^m\sum_{l}q(l,m)\Pr(Y_K=1|L_0=l,A_0=0)}\frac{1}{\mathrm{Odds}(A_0=1|Y_K=1,M>0)}\tag{by \HomogeneousRiskRatios{}}\\
&\qquad=\frac{\theta}{\mathrm{Odds}(A_0=1|Y_K=1,M>0)}
\end{align*}
where $q(l,m)=\Pr(M=m|L_0=l,Y_K=1)\Pr(A_0=0|L_0=l)\Pr(A_0=1|L_0=l)\Pr(L_0=l)$.

It follows that \begin{align*}
\frac{\Exp\big[\sum_{i=1}^M I(A'_i=0,A_0=1)\big|Y_K=1,M>0\big]}{\Exp\big[\sum_{i=1}^M I(A'_i=1,A_0=0)\big|Y_K=1,M>0\big]} &=
\frac{\Pr(Y_K=1|L_0,A_0=1)}{\Pr(Y_K=1|L_0,A_0=0)}\\
&=\frac{\Pr(Y_K(1)=1|L_0,A_0=1)}{\Pr(Y_K(0)=1|L_0,A_0=0)}\tag{by consistency}\\
&=\frac{\Pr(Y_K(1)=1|L_0)}{\Pr(Y_K(0)=1|L_0)}\tag{by baseline conditional exchangeability}\\
&=\frac{\Pr(Y_K(1)=1)}{\Pr(Y_K(0)=1)}.
\end{align*}
\end{proof}

\begin{theorem}[\titleSurvivorSamplingMatchingITT]\label{th:SurvivorSamplingMatchingITT}
Suppose \SurvivorSamplingMatching{} and BCE hold. If 
\begin{align*}
\frac{\mathrm{Odds}(Y_K=1|L_0,A_0=1)}{\mathrm{Odds}(Y_K=1|L_0,A_0=0)}=\theta \tag{\HomogeneousOddsRatios{}}
\end{align*}
for some constant $\theta$, then \begin{align*}
\frac{\Exp\big[\sum_{i=1}^M I(A'_i=0,A_0=1)\big|Y_K=1,M>0\big]}{\Exp\big[\sum_{i=1}^M I(A'_i=1,A_0=0)\big|Y_K=1,M>0\big]} = \frac{\mathrm{Odds}(Y_K(1)=1|L_0)}{\mathrm{Odds}(Y_K(0)=1|L_0)}.
\end{align*}
\end{theorem}
\begin{proof}
We have
\begin{align*}
\frac{\Exp\big[\sum_{i=1}^M I(A'_i=0,A_0=1)\big|Y_K=1,M>0\big]}{\Exp\big[\sum_{i=1}^M I(A'_i=1,A_0=0)\big|Y_K=1,M>0\big]} &= \frac{\Exp\big[\sum_{i=1}^M I(A'_i=0)\big|A_0=1,Y_K=1,M>0\big]}{\Exp\big[\sum_{i=1}^M I(A'_i=1)\big|A_0=0,Y_K=1,M>0\big]}\\&\qquad{}\times\mathrm{Odds}(A_0=1|Y_K=1,M>0),
\end{align*}
where
\begin{align*}
&\frac{\Exp\big[\sum_{i=1}^M I(A'_i=0)\big|A_0=1,Y_K=1,M>0\big]}{\Exp\big[\sum_{i=1}^M I(A'_i=1)\big|A_0=0,Y_K=1,M>0\big]}\\
&\qquad=\frac{\sum_{m>0}\Exp\big[\sum_{i=1}^m I(A'_i=0)\big|A_0=1,Y_K=1,M=m\big]\Pr(M=m|A_0=1,Y_K=1)}{\sum_{m>0}\Exp\big[\sum_{i=1}^m I(A'_i=1)\big|A_0=0,Y_K=1,M=m\big]\Pr(M=m|A_0=0,Y_K=1)}\\
&\qquad=\frac{\sum_{m>0}\sum_{i=1}^m\sum_{l}\Pr(A'_i=0|L_0=l,A_0=1,Y_K=1,M=m)\Pr(M=m,L_0=l|A_0=1,Y_K=1,M>0)}{\sum_{m>0}\sum_{i=1}^m\sum_{l}\Pr(A'_i=1|L_0=l,A_0=0,Y_K=1,M=m)\Pr(M=m,L_0=l|A_0=0,Y_K=1,M>0)} \\
&\qquad=\frac{\sum_{m>0}\sum_{i=1}^m\sum_{l}\Pr(A_0=0|L_0=l,Y_K=0)\Pr(M=m,L_0=l|A_0=1,Y_K=1,M>0)}{\sum_{m>0}\sum_{i=1}^m\sum_{l}\Pr(A_0=1|L_0=l,Y_K=0)\Pr(M=m,L_0=l|A_0=0,Y_K=1,M>0)}\tag{by \SurvivorSamplingMatching{}}\\
&\qquad=\frac{\displaystyle\sum_{m>0}\sum_{i=1}^m\sum_{l}\frac{\Pr(Y_K=0|L_0=0,A_0=0)\Pr(A_0=0|L_0=l)}{\Pr(Y_K=0|L_0=l)}\Pr(M=m,L_0=l,A_0=1|Y_K=1)}{\displaystyle\sum_{m>0}\sum_{i=1}^m\sum_{l}\frac{\Pr(Y_K=0|L_0=0,A_0=1)\Pr(A_0=1|L_0=l)}{\Pr(Y_K=0|L_0=l)}\Pr(M=m,L_0=l,A_0=0|Y_K=1)}\\&\qquad\qquad\times\frac{1}{\mathrm{Odds}(A_0=1|Y_K=1,M>0)}\\
&\qquad=\frac{\sum_{m>0}\sum_{i=1}^m\sum_{l}q(l,m)\Pr(Y_K=1|L_0=l,A_0=1)\Pr(Y_K=0|L_0=0,A_0=0)}{\sum_{m>0}\sum_{i=1}^m\sum_{l}q(l,m)\Pr(Y_K=1|L_0=l,A_0=0)\Pr(Y_K=0|L_0=0,A_0=1)}\\&\qquad\qquad\times\frac{1}{\mathrm{Odds}(A_0=1|Y_K=1,M>0)} \tag{under \SurvivorSamplingMatching{} and definition of $q(l,m)$ (see below)}\\
&\qquad=\frac{\sum_{m>0}\sum_{i=1}^m\sum_{l}q(l,m)\theta\Pr(Y_K=1|L_0=l,A_0=0)\Pr(Y_K=0|L_0=0,A_0=1)}{\sum_{m>0}\sum_{i=1}^m\sum_{l}q(l,m)\Pr(Y_K=1|L_0=l,A_0=0)\Pr(Y_K=0|L_0=0,A_0=1)}\\&\qquad\qquad\times\frac{1}{\mathrm{Odds}(A_0=1|Y_K=1,M>0)}\tag{by \HomogeneousOddsRatios{}}\\
&\qquad=\frac{\theta}{\mathrm{Odds}(A_0=1|Y_K=1,M>0)}
\end{align*}
where $q(l,m)=\Pr(M=m|L_0=l,Y_K=1)\Pr(A_0=0|L_0=l)\Pr(A_0=1|L_0=l)\Pr(L_0=l)/\Pr(Y_K=0|L_0=l)$.

From the definition of $\theta$, it follows that \begin{align*}
\frac{\Exp\big[\sum_{i=1}^M I(A'_i=0,A_0=1)\big|Y_K=1,M>0\big]}{\Exp\big[\sum_{i=1}^M I(A'_i=1,A_0=0)\big|Y_K=1,M>0\big]}
&=\frac{\mathrm{Odds}(Y_K(1)=1|L_0,A_0=1)}{\mathrm{Odds}(Y_K(0)=1|L_0,A_0=0)}\tag{by consistency}\\
&=\frac{\mathrm{Odds}(Y_K(1)=1|L_0)}{\mathrm{Odds}(Y_K(0)=1|L_0)}\tag{by baseline conditional exchangeability}\\
&=\frac{\mathrm{Odds}(Y_K(1)=1)}{\mathrm{Odds}(Y_K(0)=1)}.
\end{align*}
\end{proof}

\begin{theorem}[\titleDensitySamplingMatchingITT]\label{th:DensitySamplingMatchingITT}
Suppose \RiskSetSamplingMatching{} and BCE hold. If 
\begin{align*}
\frac{\Pr(Y_{j+1}=1|L_0,A_0=1,Y_j=0)}{\Pr(Y_{j+1}=1|L_0,A_0=0,Y_j=0)}=\theta \tag{\HomogeneousRateRatios{}}
\end{align*}
for $j=0,1,...,K$ and some constant $\theta$, then \begin{align*}
\frac{\Exp\big[\sum_{i=1}^M I(A'_i=0,A_0=1)\big|Y_K=1,M>0\big]}{\Exp\big[\sum_{i=1}^M I(A'_i=1,A_0=0)\big|Y_K=1,M>0\big]} = \frac{\Pr(Y_{j+1}(1)=1|L_0,Y_j(1)=0)}{\Pr(Y_{j+1}(0)=1|L_0,Y_j(0)=0)}.
\end{align*}
\end{theorem}
\begin{proof}
If $J=\max\{k=0,1,...,K:Y_k=0\}$, then
\begin{align*}
&\frac{\Exp\big[\sum_{i=1}^M I(A'_i=0,A_0=1)\big|Y_K=1,M>0\big]}{\Exp\big[\sum_{i=1}^M I(A'_i=1,A_0=0)\big|Y_K=1,M>0\big]}\\
&\qquad=\frac{\sum_{m>0}\Exp\big[\sum_{i=1}^m I(A'_i=0,A_0=1)\big|Y_K=1,M=m\big]\Pr(M=m|Y_K=1,M>0)}{\sum_{m>0}\Exp\big[\sum_{i=1}^m I(A'_i=1,A_0=0)\big|Y_K=1,M=m\big]\Pr(M=m|Y_K=1,M>0)}\\
&\qquad=\frac{\sum_{m>0}\Exp\big[\sum_{i=1}^m I(A'_i=0,A_0=1)\big|Y_K=1,M=m\big]\Pr(M=m,Y_K=1)}{\sum_{m>0}\Exp\big[\sum_{i=1}^m I(A'_i=1,A_0=0)\big|Y_K=1,M=m\big]\Pr(M=m,Y_K=1)}\\
&\qquad=\frac{\sum_{m>0}\sum_{j=0}^{K-1}\sum_{l}\Exp\big[\sum_{i=1}^mI(A'_i=0,A_0=1)\big|L_0=l,J=j,M=m\big]\Pr(L_0=l,J=j,M=m)}{\sum_{m>0}\sum_{j=0}^{K-1}\sum_{l}\Exp\big[\sum_{i=1}^mI(A'_i=1,A_0=0)\big|L_0=l,J=j,M=m\big]\Pr(L_0=l,J=j,M=m)} \\
&\qquad=\frac{\sum_{m>0}\sum_{i=1}^m\sum_{j=0}^{K-1}\sum_{l}\Exp\big[I(A'_i=0,A_0=1)\big|L_0=l,J=j,M=m\big]\Pr(L_0=l,J=j,M=m)}{\sum_{m>0}\sum_{i=1}^m\sum_{j=0}^{K-1}\sum_{l}\Exp\big[ I(A'_i=1,A_0=0)\big|L_0=l,J=j,M=m\big]\Pr(L_0=l,J=j,M=m)}\\
&\qquad=\frac{\displaystyle\sum_{m>0}\sum_{i=1}^m\sum_{j=0}^{K-1}\sum_{l}\begin{array}{l}\Exp\big[I(A'_i=0,A_0=1)\big|L_0=l,Y_j=0,Y_{j+1}=1,M=m\big]\\\qquad{}\times\Pr(L_0=l,Y_j=0,Y_{j+1}=1,M=m)\end{array}}{\displaystyle\sum_{m>0}\sum_{i=1}^m\sum_{j=0}^{K-1}\sum_{l}\begin{array}{l}\Exp\big[ I(A'_i=1,A_0=0)\big|L_0=l,Y_j=0,Y_{j+1}=1,M=m\big]\\\qquad{}\times\Pr(L_0=l,Y_j=0,Y_{j+1}=1,M=m)\end{array}}\\
&\qquad=\frac{\displaystyle\sum_{m>0}\sum_{i=1}^m\sum_{j=0}^{K-1}\sum_{l}\begin{array}{l}\Pr(A'_i=0|L_0=l,A_0=1,Y_j=0,Y_{j+1}=1,M=m)\\\qquad{}\times\Pr(L_0=l,A_0=1,Y_j=0,Y_{j+1}=1,M=m)\end{array}}{\displaystyle\sum_{m>0}\sum_{i=1}^m\sum_{j=0}^{K-1}\sum_{l}\begin{array}{l}\Pr(A'_i=1|L_0=l,A_0=0,Y_j=0,Y_{j+1}=1,M=m)\\\qquad{}\times\Pr(L_0=l,A_0=0,Y_j=0,Y_{j+1}=1,M=m)\end{array}}\\
&\qquad=\frac{\sum_{m>0}\sum_{i=1}^m\sum_{j=0}^{K-1}\sum_{l}\Pr(A_0=0|L_0=l,Y_j=0)\Pr(L_0=l,A_0=1,Y_j=0,Y_{j+1}=1,M=m)}{\sum_{m>0}\sum_{i=1}^m\sum_{j=0}^{K-1}\sum_{l}\Pr(A_0=1|L_0=l,Y_j=0)\Pr(L_0=l,A_0=0,Y_j=0,Y_{j+1}=1,M=m)} \tag{by \RiskSetSamplingMatching{}}\\
&\qquad=\frac{\sum_{m>0}\sum_{i=1}^m\sum_{j=0}^{K-1}\sum_l q_j(l,m)\Pr(Y_{j+1}=1|L_0=l,A_0=1,Y_j=0)}{\sum_{m>0}\sum_{i=1}^m\sum_{j=0}^{K-1}\sum_lq_j(l,m)\Pr(Y_{j+1}=1|L_0=l,A_0=0,Y_j=0)} \tag{under \RiskSetSamplingMatching{} and definition of $q_j(l,m)$ (see below)}\\
&\qquad=\theta\frac{\sum_{m>0}\sum_{i=1}^m\sum_{j=0}^{K-1}\sum_l q_j(l,m)\Pr(Y_{j+1}=1|L_0=l,A_0=0,Y_j=0)}{\sum_{m>0}\sum_{i=1}^m\sum_{j=0}^{K-1}\sum_lq_j(l,m)\Pr(Y_{j+1}=1|L_0=l,A_0=0,Y_j=0)}\tag{by \HomogeneousRateRatios{}}\\
&\qquad=\theta.
\end{align*}
where $q_{j}(l,m)=\Pr(M=m|L_0=l,Y_j=0)\Pr(A_0=1|L_0=l,Y_j=0)\Pr(A_0=0|L_0=l,Y_j=0)\Pr(L_0=l,Y_j=0)$.

Thus,
\begin{align*}
\frac{\Exp\big[\sum_{i=1}^M I(A'_i=0,A_0=1)\big|Y_K=1\big]}{\Exp\big[\sum_{i=1}^M I(A'_i=1,A_0=0)\big|Y_K=1\big]}&=\frac{\Pr(Y_{j+1}=1|L_0,A_0=1,Y_j=0)}{\Pr(Y_{j+1}=1|L_0,A_0=0,Y_j=0)}\\
&=\frac{\Pr(Y_{j+1}(1)=1|L_0,A_0=1,Y_j(1)=0)}{\Pr(Y_{j+1}(0)=1|L_0,A_0=0,Y_j(0)=0)}\tag{by consistency}\\
&=\frac{\Pr(Y_{j+1}(1)=1|L_0,Y_j(1)=0)}{\Pr(Y_{j+1}(0)=1|L_0,Y_j(0)=0)}.\tag{by baseline conditional exchangeability}
\end{align*}
\end{proof}

\subsection*{Per-protocol effect}

In this subsection, an individual qualifies as a case if and only if $Y_K=1$ and the subject adheres to the protocol that was assigned at baseline (i.e., $A_k=A_0$ for all $k=0,1,...,K$ if $Y_k=0$). All cases are assigned a (possibly variable) number $M\ge0$ control exposures $A'_i$, $i=1,...,M$, subject to 
\begin{align*}
\left.\begin{array}{c}
\Pr(M>0|Y_K=1,\forall j:(Y_j=0\Rightarrow A_j=A_0))>0\text{~and}\\
M\CI A_0|(J,Y_K=1,\overline{L}_J=\overline{l}_J,\forall{i\le J}:A_i=A_0)\text{~and}\\
%{\color{red}A_0,A'_1,...,A'_M~{\text{mutually independent given}}~(L_0,\overline{Y}_K,M)} \text{~and}\\
\forall{\overline{l},a}: \Pr(A'_i=a'|\overline{L}_J=\overline{l}_J,\forall j\le J:A_{j}=A_0,A_0=a,Y_J=0,J,M,M>0)\\\qquad= \Pr(A_J=a'|\overline{L}_J=\overline{l}_J,\forall j\le J:A_{j}=A_0,Y_J=0),~~\text{where}\\\qquad\qquad J=\max\{k=0,1,...,K:Y_k=0\}.\end{array}\right\} \tag{\PPRiskSetSamplingMatching{}}
\end{align*}
% {\color{red}or \begin{align*}
% \forall{\overline{l},a}: &\Pr(A'=a'|\overline{L}_J=\overline{l}_J,\forall j\le J:A_{j}=A_0,A_0=a,Y_J=0,J)\\&\qquad= \Pr(A_J=a'|\overline{L}_J=\overline{l}_J,\forall j\le J:A_{j}=A_0,Y_K=0),~~\text{where}\\&\qquad\qquad J=\max\{k=0,1,...,K:Y_k=0\}. \tag{M5}
% \end{align*}
% }

% {\color{red}
% \begin{theorem}[{\color{orange}Case-cohort} sampling with matching for per-protocol effect]
% ... % risk in the always exposed divided by risk in never exposed
% \begin{align*}
% \frac{\Pr(A'=0,A_0=1|Y_K=1,\forall j: (Y_j=0\Rightarrow A_j=A_0))}{\Pr(A'=1,A_0=0|Y_K=1,\forall j: (Y_j=0\Rightarrow A_j=A_0))}
% \end{align*}
% \end{theorem}
% \begin{proof}
% Let $J=\max\{k=0,1,...,K:Y_k=0\}$. Then, for $a=0,1$,
% \begin{align*}
% &\Pr(A'=1-a,A_0=a|Y_K=1,\forall j:(Y_j=0 \Rightarrow A_j=A_0))\\
% %%%%
% &\qquad=\frac{\Pr(A'=1-a,Y_K=1,A_0=a,\forall j:(Y_j=0 \Rightarrow A_j=A_0))}{\Pr(Y_K=1,\forall j:(Y_j=0 \Rightarrow A_j=A_0)}\\
% %%%%
% &\qquad=\frac{\Pr(Y_K=1|A'=1-a,A_0=a,\forall j:(Y_j=0 \Rightarrow A_j=A_0))\Pr(A'=1-a,A_0=a,\forall j:(Y_j=0 \Rightarrow A_j=A_0))}{\Pr(Y_K=1,\forall j:(Y_j=0 \Rightarrow A_j=A_0)}\\
% \end{align*}
% \end{proof}
% }

\begin{theorem}[\titleDensitySamplingMatchingPP]\label{th:DensitySamplingMatchingPP}
Suppose \PPRiskSetSamplingMatching{} holds. If \begin{align*}
\frac{\Pr(Y_{j+1}=1|\overline{L}_j=\overline{l}_j,Y_{j}=0,\forall{i\le j}:A_i=1)}{\Pr(Y_{j+1}=1|\overline{L}_j=\overline{l}_j,Y_{j}=0,\forall{i\le j}:A_i=0)} = \theta \tag{\PPHomogeneousRateRatios{}}
\end{align*} for all $j,\overline{l}_j$ and some constant $\theta$, then
\begin{align*}
&\frac{\Exp\Big[\sum_{i=1}^MI(A'_i=0,A_0=1)\Big|Y_K=1,\forall j: (Y_j=0\Rightarrow A_j=A_0),M>0\Big]}{\Exp\Big[\sum_{i=1}^MI(A'_i=1,A_0=0)\Big|Y_K=1,\forall j: (Y_j=0\Rightarrow A_j=A_0),M>0\Big]} \\&\qquad= \frac{\Pr(Y_{j+1}(\overline{1})=1|\overline{L}_j=\overline{l}_j,Y_{j}(\overline{1})=0,\forall{i\le j}:A_i=1)}{\Pr(Y_{j+1}(\overline{0})=1|\overline{L}_j=\overline{l}_j,Y_{j}(\overline{0})=0,\forall{i\le j}:A_i=0)}.
\end{align*}
\end{theorem}
\begin{proof}
Let $J=\max\{k=0,1,...,K:Y_k=0\}$. Then, for $a=0,1$,
\begin{align*}
&\Exp\Bigg[\sum_{i=1}^MI(A'_i=1-a,A_0=a)\Bigg|Y_K=1,\forall j\le J:A_j=A_0, M>0\Bigg]\\
&\qquad=\sum_{j=0}^{K-1}\sum_{\overline{l}_j}\Exp\Bigg[\sum_{i=1}^MI(A'_i=1-a,A_0=a)\Bigg|\overline{L}_j=\overline{l}_j,J=j,Y_K=1,\forall j\le J:A_j=A_0, M>0\Bigg]\\&\qquad\qquad{}\times\Pr(\overline{L}_j=\overline{l}_j,J=j|Y_K=1,\forall{i\le J}:A_i=A_0,M>0)\\
&\qquad=\sum_{j=0}^{K-1}\sum_{\overline{l}_j}\Exp\Bigg[\sum_{i=1}^MI(A'_i=1-a,A_0=a)\Bigg|\overline{L}_j=\overline{l}_j,Y_j=0,Y_{j+1}=1,\forall j\le J:A_j=A_0, M>0\Bigg]\\&\qquad\qquad{}\times\Pr(\overline{L}_j=\overline{l}_j,Y_j=0,Y_{j+1}=1|Y_K=1,\forall{i\le J}:A_i=A_0,M>0)\\
&\qquad=\sum_{m>0}\sum_{j=0}^{K-1}\sum_{\overline{l}_j}\Exp\Bigg[\sum_{i=1}^mI(A'_i=1-a,A_0=a)\Bigg|\overline{L}_j=\overline{l}_j,Y_j=0,Y_{j+1}=1,\forall j\le J:A_j=A_0, M=m\Bigg]\\&\qquad\qquad{}\times\Pr(M=m,\overline{L}_j=\overline{l}_j,Y_j=0,Y_{j+1}=1|\forall j\le J:A_j=A_0,M>0)\\
&\qquad=\sum_{m>0}\sum_{u=1}^m\sum_{j=0}^{K-1}\sum_{\overline{l}_j}\Exp\Bigg[I(A'_u=1-a,A_0=a)\Bigg|\overline{L}_j=\overline{l}_j,Y_j=0,Y_{j+1}=1,\forall j\le J:A_j=A_0, M=m\Bigg]\\&\qquad\qquad{}\times\Pr(M=m,\overline{L}_j=\overline{l}_j,Y_j=0,Y_{j+1}=1|\forall j\le J:A_j=A_0,M>0)\\
&\qquad=\sum_{m>0}\sum_{u=1}^m\sum_{j=0}^{K-1}\sum_{\overline{l}_j}\Pr(A'_u=1-a|\overline{L}_j=\overline{l}_j,Y_j=0,Y_{j+1}=1,\forall j\le J:A_j=a, M=m)\\&\qquad\qquad{}\times\Pr(M=m,A_0=a,\overline{L}_j=\overline{l}_j,Y_j=0,Y_{j+1}=1|\forall j\le J:A_j=A_0,M>0)\\
&\qquad=\sum_{m>0}\sum_{u=1}^m\sum_{j=0}^{K-1}\sum_{\overline{l}_j}\Pr(A_0=1-a|Y_j=0,\overline{L}_j=\overline{l}_j,\forall i\le j:A_i=A_0)\\&\qquad\qquad{}\times\Pr(M=m,A_0=a,\overline{L}_j=\overline{l}_j,Y_j=0,Y_{j+1}=1|\forall j\le J:A_j=A_0,M>0) \tag{by \PPRiskSetSamplingMatching{}}\\
&\qquad=\sum_{m>0}\sum_{u=1}^m\sum_{j=0}^{K-1}\sum_{\overline{l}_j}\Pr(A_0=1-a|Y_j=0,\overline{L}_j=\overline{l}_j,\forall i\le j:A_i=A_0)\\&\qquad\qquad\times\Pr(M=m,\overline{L}_j=\overline{l}_j,A_0=a,Y_{j+1}=1,Y_{j}=0,\forall{i\le j}:A_i=A_0)\\&\qquad\qquad\times\Pr(Y_K=1,\forall{i}:(Y_i=0\Rightarrow A_i=A_0),M>0)^{-1}\\
%%%%
&\qquad=\sum_{m>0}\sum_{u=1}^m\sum_{j=0}^{K-1}\sum_{\overline{l}_j}
\Pr(Y_{j+1}=1|\overline{L}_j=\overline{l}_j,A_0=a,Y_{j}=0,\forall{i\le j}:A_i=A_0)\\
&\qquad\qquad\times
q_{j}(\overline{l}_j,m)
\Pr(Y_K=1,\forall{i}:(Y_i=0\Rightarrow A_i=A_0),M>0)^{-1}, \tag{under \PPRiskSetSamplingMatching{}}
\end{align*}
where \begin{align*}
q_{j}(\overline{l}_j,m)&=
\Pr(M=m|\overline{L}_j=\overline{l}_j,Y_{j}=0,Y_{j+1}=1,\forall{i\le j}:A_i=A_0) \\
&\qquad\times \Pr(A_0=1-a|Y_j=0,\overline{L}_j=\overline{l}_j,\forall i\le j:A_i=A_0)\\
&\qquad\times \Pr(A_0=a|Y_{j}=0,\overline{L}_j=\overline{l}_j,\forall{i\le j}:A_i=A_0)\\
&\qquad\times \Pr(\overline{L}_j=\overline{l}_j,Y_{j}=0,\forall{i\le j}:A_i=A_0).
\end{align*}

It follows that
\begin{align*}
&\frac{\Exp\Big[\sum_{i=1}^MI(A'_i=0,A_0=1)\Big|Y_K=1,\forall j: (Y_j=0\Rightarrow A_j=A_0),M>0\Big]}{\Exp\Big[\sum_{i=1}^MI(A'_i=1,A_0=0)\Big|Y_K=1,\forall j: (Y_j=0\Rightarrow A_j=A_0),M>0\Big]}\\
&\qquad=\frac{\displaystyle \sum_{m>0}\sum_{u=1}^m\sum_{j=0}^{K-1}\sum_{\overline{l}_j}\begin{array}{l} \Pr(Y_{j+1}=1|\overline{L}_j=\overline{l}_j,A_0=1,Y_{j}=0,\forall{i\le j}:A_i=A_0)\\~~\times
q_{j}(\overline{l}_j,m)
\Pr(Y_K=1,\forall{i}:(Y_i=0\Rightarrow A_i=A_0),M>0)^{-1}\end{array}}
{\displaystyle \sum_{m>0}\sum_{u=1}^m\sum_{j=0}^{K-1}\sum_{\overline{l}_j}\begin{array}{l} \Pr(Y_{j+1}=1|\overline{L}_j=\overline{l}_j,A_0=0,Y_{j}=0,\forall{i\le j}:A_i=A_0)\\~~\times
q_{j}(\overline{l}_j,m)
\Pr(Y_K=1,\forall{i}:(Y_i=0\Rightarrow A_i=A_0),M>0)^{-1}\end{array}} \\
&\qquad = \frac{\sum_{m>0}\sum_{u=1}^m\sum_{j=0}^{K-1}\sum_{\overline{l}_j}\Pr(Y_{j+1}=1|\overline{L}_j=\overline{l}_j,A_0=1,Y_{j}=0,\forall{i\le j}:A_i=A_0)
q_{j}(\overline{l}_j,m)}
{\sum_{m>0}\sum_{u=1}^m\sum_{j=0}^{K-1}\sum_{\overline{l}_j}\Pr(Y_{j+1}=1|\overline{L}_j=\overline{l}_j,A_0=1,Y_{j}=0,\forall{i\le j}:A_i=A_0)
q_{j}(\overline{l}_j,m)} \\
&\qquad = \theta\frac{\sum_{m>0}\sum_{u=1}^m\sum_{j=0}^{K-1}\sum_{\overline{l}_j}\Pr(Y_{j+1}=1|\overline{L}_j=\overline{l}_jY_{j}=0,\forall{i\le j}:A_i=0)
q_{j}(\overline{l}_j,m)}
{\sum_{m>0}\sum_{u=1}^m\sum_{j=0}^{K-1}\sum_{\overline{l}_j}\Pr(Y_{j+1}=1|\overline{L}_j=\overline{l}_j,Y_{j}=0,\forall{i\le j}:A_i=0)
q_{j}(\overline{l}_j,m)} \tag{by \PPHomogeneousRateRatios{}}\\
&\qquad = \theta.
\end{align*}
The desired results follows by consistency.
\end{proof}

\end{myAppendix}

\begin{myAppendix}[: Parametric identification by conditional logistic regression for exact or partial 1:$M$ matching] \label{app:MatchingPartial}
We now allow for the possibility that cases ($Y_K=1$) are matched to $M\ge 0$ controls on only part of $L_0$. That part of $L_0$ on which exact matching is done will be denoted $L^\ast_0$; the other part is denoted $L'_0$, so that $L_0=(L^\ast_0,L'_0)$.
The identification result below rests on the assumption that cases are assigned $M\ge 0$ pairs $(A'_i,L'_i)$ of baseline exposure and baseline covariate data, $i=1,...,M$, subject to
\begin{gather*}
\left.\begin{array}{c}
\Pr(M>0|Y_K=1)>0\text{~~and}\\
M\CI (A_0,L_0)|(L_0^\ast,Y_K=1)\text{~~and}\\
\forall l,l',a:\Pr(A'_i=a,L'_i=l'|L_0^\ast=l,L'_0,A_0,Y_K=1,M,M>0) = \qquad\qquad\qquad\qquad\\\qquad\qquad\qquad\qquad\Pr(A_0=a,L'_0=l'|L_0^\ast=l,Y_K=0)\text{~~and}\\
(L'_0,A_0),(L'_1,A'_1),...,(L'_M,A'_M)\text{~are mutually independent given~} (L_0^\ast,Y_K=1,M>0).
\end{array}\right\}\tag{\SurvivorSamplingMatching{}$^\ast$}
\end{gather*}
It is assumed below that the variables are discrete with finite support for simplicity. The results can however be extended to more general distributions.

\begin{theorem}[\titleConditionalLogisticRegression]\label{th:ConditionalLogisticRegression}
Suppose BCE and \SurvivorSamplingMatching{}$^\ast$ hold.
For some known real-valued functions $f_j$, $j=1,...,p$, assume the following model:
\begin{align*}
\logit\Pr(Y_K(a)=1|L_0) &= \alpha+ \sum_{j=1}^pf_j(a,L_0^\ast,L'_0)\beta_j \tag{Outcome Model}
\end{align*}
For $i=0,...,M$, let $X_{i,j}=f_j(A'_i,L_0^\ast,L'_i)-f_j(A_0,L_0^\ast,L'_0)$, with $A'_0=A_0$, and assume for any $\gamma_1,...,\gamma_p\in\mathbb{R}$, not all zero, that
\begin{align*}
%\sum_{i=1}^M\sum_{j=1}^{p}I(\gamma_{j}\ne 0)>0\Rightarrow 
\Pr\Bigg(\bigvee_{i=1}^M\Bigg[\sum_{j=1}^p \gamma_{j} X_{i,j}\ne 0\Bigg]\Bigg|Y_K=1,M>0\Bigg)>0,\tag{Linear Independence}
\end{align*}
where $\bigvee$ denotes the logical OR operator (i.e., given any indexed collection $(P_i)_{i\in I}$ of propositions, $\bigvee_{i\in I}P_i$ is the proposition that $P_i$ is true for at least one $i\in I$).
Then,
\begin{align*}
\Exp\Bigg[-\log \Bigg(1+\sum_{i=1}^M\exp\Bigg[\sum_{j=1}^pX_{i,j}\tilde{\beta}_j\Bigg]\Bigg)^{-1}\Bigg|Y_K=1,M>0\Bigg]
\end{align*}
is uniquely maximized at $\tilde{\beta}=\beta$.
\end{theorem}
\begin{proof}
We first demonstrate that \begin{align*}
\Exp\Bigg[-\log \Bigg(1+\sum_{i=1}^M\exp\Bigg[\sum_{j=1}^pX_{i,j}\tilde{\beta}_j\Bigg]\Bigg)^{-1}\Bigg|Y_K=1,M>0\Bigg]
\end{align*}
has at most one maximum by showing that it is strictly concave as a function of $\tilde{\beta}$. 
Let $X=(X_1,...,X_M)$ and $X_i=(X_{i,1},...,X_{i,p})$, $i=1,...,M$. To show that function $f$,
\begin{align*}
f(\beta) &= \Exp\Bigg[\log \Bigg(1+\sum_{i=1}^M\exp\Bigg[\sum_{j=1}^pX_{i,j}\beta_j\Bigg]\Bigg)^{-1}\Bigg|Y_K=1,M>0\Bigg]\\
&= \sum_{m>0}\sum_{x} \log \Bigg(1+\sum_{i=1}^m\exp\Bigg[\sum_{j=1}^px_{i,j}\beta_j\Bigg]\Bigg)^{-1} \Pr(X=x|Y_K=1,M=m)\Pr(M=m|Y_K=1,M>0),
\end{align*}
is convex (and $-f$ concave) it suffices to show that its Hessian is positive semidefinite, i.e., that $\sum_{t=1}^p\sum_{u=1}^p \beta_k\beta_lH_{k,l}(\beta)\ge 0$ for all $\beta\in\mathbb{R}^p$, where 
\begin{align*}
H_{k,l}(\beta) &= \frac{\partial}{\partial \beta_l}\frac{\partial}{\partial \beta_k} f(\beta).
\end{align*}
Positive definiteness of the Hessian, i.e., $\sum_{k=1}^p\sum_{l=1}^p \beta_k\beta_lH_{k,l}(\beta)> 0$ for all $\beta\in\mathbb{R}^p$ such that $\beta_k\ne 0$ for some $k\in\{1,...,p\}$, implies strict convexity of $f$ (and $-f$ strictly concave).

Letting $g(X_i,\beta)=\exp\big\{\sum_{j=1}^pX_{i,j}\beta_j\big\}$ for $i=1,...,M$, we have
\begin{align*}
H_{k,l}(\beta) &= \frac{\partial}{\partial \beta_l}\frac{\partial}{\partial \beta_k} f(\beta)\\
&= \frac{\partial}{\partial \beta_l} \sum_{m>0}\sum_{x} \frac{\sum_{i=1}^{m}x_{i,k}g(x_i,\beta)}{1+\sum_{i=1}^{m}g(x_i,\beta)} \Pr(X=x|Y_K=1,M=m)\Pr(M=m|Y_K=1,M>0)\\
&= \frac{\partial}{\partial \beta_l} \sum_{m>0}\sum_{x} \frac{\sum_{i=1}^{m}x_{i,k}g(x_i,\beta)}{1+\sum_{i=1}^{m}g(x_i,\beta)} \Pr(X=x|Y_K=1,M=m)\Pr(M=m|Y_K=1,M>0)\\
&= \sum_{m>0}\sum_{x} \Bigg(1+\sum_{i=1}^{m}g(x_i,\beta)\Bigg)^{-2}\Bigg[\Bigg(1+\sum_{i=1}^{m}g(x_i,\beta)\Bigg)\Bigg(\sum_{i=1}^{m}X_{i,k}X_{i,l}g(x_i,\beta)\Bigg)\\&\qquad{}-\Bigg(\sum_{i=1}^{m}X_{i,k}g(x_i,\beta)\Bigg)\Bigg(\sum_{i=1}^{m}X_{i,l}g(x_i,\beta)\Bigg)\Bigg] \\&\qquad{}\times\Pr(X=x|Y_K=1,M=m)\Pr(M=m|Y_K=1,M>0),
\end{align*}
so that, with $v_{i}=\sqrt{g(x_i,\beta)}$ and $w_{i}=\sum_{j=1}^px_{i,j}\beta_j\sqrt{g(x_i,\beta)}$, 
\begin{align*}
&\sum_{k=1}^p\sum_{l=1}^p\beta_k\beta_l H_{k,l}(\beta) \\&\qquad= \sum_{m>0}\sum_{x}\frac{\Pr(X=x|Y_K=1,M=m)\Pr(M=m|Y_K=1,M>0)}{\big(1+\sum_{i=1}^mg(x_i,\beta)\big)^2}\\&\qquad\qquad{}\times\Bigg[\sum_{k=1}^p\sum_{l=1}^p\beta_k\beta_l\Bigg(1+\sum_{i=1}^mg(x_i,\beta)\Bigg)\Bigg(\sum_{i=1}^mx_{i,k}x_{i,l}g(x_i,\beta)\Bigg)\\&\qquad\qquad{}-\sum_{k=1}^p\sum_{l=1}^p\beta_k\beta_l\Bigg(\sum_{i=1}^mx_{i,k}g(x_i,\beta)\Bigg)\Bigg(\sum_{i=1}^mx_{i,l}g(x_i,\beta)\Bigg)\Bigg]\\
&\qquad=\sum_{m>0}\sum_{x}\frac{\Pr(X=x|Y_K=1,M=m)\Pr(M=m|Y_K=1,M>0)}{\big(1+\sum_{i=1}^mg(x_i,\beta)\big)^2}\\&\qquad\qquad{}\times\Bigg[\Bigg(1+\sum_{i=1}^mg(x_i,\beta)\Bigg)\Bigg(\sum_{i=1}^mg(x_i,\beta)\Bigg(\sum_{k=1}^p\beta_kx_{i,k}\Bigg)\Bigg(\sum_{l=1}^p\beta_lx_{i,l}\Bigg)\Bigg)\\&\qquad\qquad{}-\Bigg(\sum_{i=1}^m\sum_{k=1}^p\beta_kx_{i,k}g(x_i,\beta)\Bigg)\Bigg(\sum_{i=1}^m\sum_{l=1}^p\beta_lx_{i,l}g(x_i,\beta)\Bigg)\Bigg]\\
&\qquad=\sum_{m>0}\sum_{x}\frac{\Pr(X=x|Y_K=1,M=m)\Pr(M=m|Y_K=1,M>0)}{\big(1+\sum_{i=1}^mg(x_i,\beta)\big)^2}\\&\qquad\qquad{}\times\Bigg[\Bigg(1+\sum_{i=1}^mg(x_i,\beta)\Bigg)\Bigg(\sum_{i=1}^m\Bigg(\sum_{k=1}^p\beta_kx_{i,k}\sqrt{g(x_i,\beta)}\Bigg)^2\Bigg)-\Bigg(\sum_{i=1}^m\sum_{k=1}^p\beta_kx_{i,k}g(x_i,\beta)\Bigg)^2\Bigg]\\
&\qquad=\sum_{m>0}\sum_{x}\frac{\Pr(X=x|Y_K=1,M=m)\Pr(M=m|Y_K=1,M>0)}{\big(1+\sum_{i=1}^mg(x_i,\beta)\big)^2}\\&\qquad\qquad{}\times\Bigg[\sum_{i=1}^m\Bigg(\sum_{k=1}^p\beta_kx_{i,k}\sqrt{g(x_i,\beta)}\Bigg)^2+\Bigg(\sum_{i=1}^mv_{i,j}^2\Bigg)\Bigg(\sum_{i=1}^mw_{i,j}^2\Bigg)-\Bigg(\sum_{i=1}^mv_{i,j}v_{i,j}\Bigg)^2\Bigg] \\
&\qquad\ge\sum_{m>0}\sum_{x}\frac{\Pr(X=x|Y_K=1,M=m)\Pr(M=m|Y_K=1,M>0)}{\big(1+\sum_{i=1}^mg(x_i,\beta)\big)^2}\sum_{i=1}^m\Bigg(\sum_{k=1}^p\beta_kx_{i,k}\sqrt{g(x_i,\beta)}\Bigg)^2. \tag{by the Cauchy-Schwarz inequality}
\end{align*}
Now, 
\begin{align*}
&\sum_{m>0}\sum_{x}\frac{\Pr(X=x|Y_K=1,M=m)\Pr(M=m|Y_K=1,M>0)}{\big(1+\sum_{i=1}^mg(x_i,\beta)\big)^2}\sum_{i=1}^m\Bigg(\sum_{k=1}^p\beta_kx_{i,k}\sqrt{g(x_i,\beta)}\Bigg)^2 \\&\qquad= \sum_{m>0}\sum_{x}\frac{\Pr(X=x|Y_K=1,M=m)\Pr(M=m|Y_K=1,M>0)}{\big(1+\sum_{i=1}^mg(x_i,\beta)\big)^2}\sum_{i=1}^mg(x_i,\beta)\Bigg(\sum_{k=1}^p\beta_kx_{i,k}\Bigg)^2\\
\\&\qquad= \Exp\Bigg[\Bigg(1+\sum_{i=1}^Mg(X_i,\beta)\Bigg)^{-2}\sum_{i=1}^Mg(X_i,\beta)\Bigg(\sum_{k=1}^p\beta_kX_{i,k}\Bigg)^2\Bigg|Y_K=1,M>0\Bigg]\\
&\qquad\ge0
\end{align*}
with strict inequality under Linear Independence. Thus, \begin{align*}
\Exp\Bigg[-\log \Bigg(1+\sum_{i=1}^M\exp\Bigg[\sum_{j=1}^pX_{i,j}\tilde{\beta}_j\Bigg]\Bigg)^{-1}\Bigg|Y_K=1,M>0\Bigg]
\end{align*}
has at most one maximum.

It remains to be shown that \begin{align*}
\Exp\Bigg[-\log \Bigg(1+\sum_{i=1}^M\exp\Bigg[\sum_{j=1}^pX_{i,j}\tilde{\beta}_j\Bigg]\Bigg)^{-1}\Bigg|Y_K=1,M>0\Bigg]
\end{align*}
is maximized at $\tilde{\beta}=\beta$, i.e., $\partial/\partial \tilde{\beta}_k f(\tilde{\beta}) = 0$ for all $k=1,...,p$ at $\tilde{\beta}=\beta$.

Now,
\begin{align*}
\frac{\partial}{\partial \tilde{\beta}_k} f(\tilde{\beta}) &= \Exp\Bigg[ \frac{\sum_{i=1}^{M}X_{i,k}g(X_i,\tilde{\beta})}{1+\sum_{i=1}^{m}g(X_i,\tilde{\beta})}\Bigg|Y_K=1,M>0\Bigg]\\
&= \sum_{l^\ast}\sum_{m>0}\Exp\Bigg[\frac{\sum_{i=1}^{m}X_{i,k}g(X_i,\tilde{\beta})}{1+\sum_{i=1}^{m}g(X_i,\tilde{\beta})}\Bigg|L^\ast_0=l^\ast,Y_K=1,M=m\Bigg]\Pr(L^\ast_0=l^\ast,M=m|,Y_K=1,M>0),
\end{align*}
where 
\begin{align*}
&\Exp\Bigg[\frac{\sum_{i=1}^{m}X_{i,k}g(X_i,\tilde{\beta})}{1+\sum_{i=1}^{m}g(X_i,\tilde{\beta})}\Bigg|L^\ast_0=l^\ast,Y_K=1,M=m\Bigg]\\
&\qquad= \sum_{l_0,...,l_m}\sum_{a_0,...,a_m}\frac{\sum_{i=1}^{m}[f_k(a_i,l^\ast,l_i)-f_k(a_0,l^\ast,l_0)]\exp\big\{\sum_{k=1}^p[f_k(a_i,l^\ast,l_i)-f_k(a_0,l^\ast,l_0)]\tilde{\beta}_k\big\}}{1+\sum_{i=1}^{m}\exp\big\{\sum_{k=1}^p[f_k(a_i,l^\ast,l_i)-f_k(a_0,l^\ast,l_0)]\tilde{\beta}_k\big\}}\\&\qquad\qquad{}\times\Pr(A_0=a_0,A'_1=a_1,...,A_m=a_m,L'_0=l_0,...,L'_m=l_m|L^\ast_0=l^\ast,Y_K=1,M=m)\\
&\qquad= \sum_{l_0,...,l_m}\sum_{a_0,...,a_m}\frac{\sum_{i=1}^{m}[f_k(a_i,l^\ast,l_i)-f_k(a_0,l^\ast,l_0)]\exp\big\{\sum_{k=1}^p[f_k(a_i,l^\ast,l_i)-f_k(a_0,l^\ast,l_0)]\tilde{\beta}_k\big\}}{1+\sum_{i=1}^{m}\exp\big\{\sum_{k=1}^p[f_k(a_i,l^\ast,l_i)-f_k(a_0,l^\ast,l_0)]\tilde{\beta}_k\big\}}\\&{}\times h(a_0,...,a_M,l_0,...,l_M)\Pr\Big(A_0=a_0,A'_1=a_1,...,A_M=a_M,L'_0=l_0,...,L'_m=l_m\Big|\\&\quad{}\bigvee_{\sigma}\big[(A_0=a_{\sigma(0)},L'_0=_{\sigma(0)},A'_1=a_{\sigma(1)},L'_1=_{\sigma(1)},...,A_m=a_{\sigma(m)},L'_m=_{\sigma(m)})\big],L^\ast_0=l^\ast,Y_K=1,M=m\Big),
\end{align*}
where permutation $\sigma$ denotes a bijection from $\{0,1,...,M\}$ to itself and \begin{align*}
&h(a_0,...,a_M,l_0,...,l_M)\\&\qquad=\Pr\Big(\bigvee_{\sigma}\big[(A_0=a_{\sigma(0)},L'_0=_{\sigma(0)},A'_1=a_{\sigma(1)},L'_1=_{\sigma(1)},...,A_m=a_{\sigma(m)},L'_m=_{\sigma(m)})\big]\Big|\\&\qquad\qquad{}L^\ast_0=l^\ast,Y_K=1,M=m\Big).
\end{align*}

Next, let $B_0=(L'_0,A_0)$ and $B_i=(L'_i,A'_i)$, $i=1,2,...,M$.  Let $b_i=(l_i,a_i)$ for $i=0,...,M$.
We have
\begin{align*}
&\Pr\Bigg(B_0=b_0,,...,B_M=b_M\Bigg|\bigvee_{\sigma}\big[(B_0,...,B_M)=(b_{\sigma(0)},...,b_{\sigma(M)})\big],L_0^\ast,Y_K=1,M,M>0\Bigg)\\
&\qquad=\frac{\Pr(B_0=b_0,...,B_M=b_M|L_0^\ast,Y_K=1,M>0)}{\Pr\Big(\bigvee_{\sigma}\big[B_0=b_{\sigma(0)},...,B_M=a_{\sigma(M)}\big]\Big|L_0^\ast,Y_K=1,M,M>0\Big)}\\
&\qquad\propto\frac{\Pr(B_0=b_0,...,B_M=b_M|L_0^\ast,Y_K=1,M>0)}{\sum_{\sigma}\Pr\Big(B_0=b_{\sigma(0)},...,B_M=a_{\sigma(M)}\Big|L_0^\ast,Y_K=1,M,M>0\Big)}\\
&\qquad=\frac{\prod_{i=0}^M\Pr(B_i=b_i|L_0^\ast,Y_K=1,M,M>0)}{\sum_{\sigma}\prod_{i=0}^M\Pr(B_i=b_{\sigma(i)}|L_0^\ast,Y_K=1,M,M>0)}\tag{by mutual independence of \SurvivorSamplingMatching{}$^\ast$}\\
&\qquad=\frac{\Pr(B_0=b_0|L_0^\ast,Y_K=1)\prod_{i=1}^M\Pr(B_0=b_i|L_0^\ast,Y_K=0)}{\sum_{\sigma}\Pr(B_0=b_{\sigma(0)}|L_0^\ast,Y_K=1)\prod_{i=1}^M\Pr(B_0=b_{\sigma(i)}|L_0^\ast,Y_K=0)}\tag{by \SurvivorSamplingMatching{}$^\ast$}\\
&\qquad=\frac{\Pr(Y_K=1|B_0=b_0,L_0^\ast)\prod_{i=1}^M[1-\Pr(Y_K=1|B_0=b_i,L_0^\ast)]}{\sum_{\sigma}\Pr(Y_K=1|B_0=b_{\sigma(0)},L_0^\ast)\prod_{i=1}^M[1-\Pr(Y_K=1|B_0=b_{\sigma(i)},L_0^\ast)]} \\
&\qquad=\frac{\Pr(Y_K=1|L_0=(L_0^\ast,l_0),A_0=a_0)\prod_{i=1}^M[1-\Pr(Y_K=1|L_0=(L_0^\ast,l_i),A_0=a_i)]}{\sum_{\sigma}\Pr(Y_K=1|L_0=(L_0^\ast,l_{\sigma(0)}),A_0=a_{\sigma(0)})\prod_{i=1}^M[1-\Pr(Y_K=1|L_0=(L_0^\ast,l_{\sigma(i)}),A_0=a_{\sigma(i)})]} \\
&\qquad=\frac{\displaystyle\frac{\Pr(Y_K=1|L_0=(L_0^\ast,l_0),A_0=a_0)}{1-\Pr(Y_K=1|L_0=(L_0^\ast,l_0),A_0=a_0)}\prod_{i=0}^M[1-\Pr(Y_K=1|L_0=(L_0^\ast,l_i),A_0=a_i)]}{\displaystyle\sum_{\sigma}\frac{\Pr(Y_K=1|L_0=(L_0^\ast,l_{\sigma(0)}),A_0=a_{\sigma(0)})}{1-\Pr(Y_K=1|L_0=(L_0^\ast,l_{\sigma(0)}),A_0=a_{\sigma(0)})}\prod_{i=0}^M[1-\Pr(Y_K=1|L_0=(L_0^\ast,l_{\sigma(i)}),A_0=a_{\sigma(i)})]} \\
&\qquad=\frac{\displaystyle\frac{\Pr(Y_K=1|L_0=(L_0^\ast,l_0),A_0=a_0)}{1-\Pr(Y_K=1|L_0=(L_0^\ast,l_0),A_0=a_0)}}{\displaystyle\sum_{\sigma}\frac{\Pr(Y_K=1|L_0=(L_0^\ast,l_{\sigma(0)}),A_0=a_{\sigma(0)})}{1-\Pr(Y_K=1|L_0=(L_0^\ast,l_{\sigma(0)}),A_0=a_{\sigma(0)})}} \\
&\qquad\propto\frac{\displaystyle\frac{\Pr(Y_K=1|L_0=(L_0^\ast,l_0),A_0=a_0)}{1-\Pr(Y_K=1|L_0=(L_0^\ast,l_0),A_0=a_0)}}{\displaystyle\sum_{i=0}^M\frac{\Pr(Y_K=1|L_0=(L_0^\ast,l_i),A_0=a_i)}{1-\Pr(Y_K=1|L_0=(L_0^\ast,l_i),A_0=a_i)}} \\
&\qquad=\frac{\displaystyle\frac{\expit\big\{\alpha+\sum_{j=1}^pf_j(a_0,L_0^\ast,l_0)\beta_j\big\}}{1-\expit\big\{\alpha+\sum_{j=1}^pf_j(a_0,L_0^\ast,l_0)\beta_j\big\}}}{\displaystyle\sum_{i=0}^M\frac{\expit\big\{\alpha+\sum_{j=1}^pf_j(a_i,L_0^\ast,l_i)\beta_j\big\}}{1-\expit\big\{\alpha+\sum_{j=1}^pf_j(a_i,L_0^\ast,l_i)\beta_j\big\}}}\\
&\qquad=\frac{\exp\big[\sum_{j=1}^pf_j(a_0,L_0^\ast,l_0)\beta_j\big]}{\sum_{i=0}^M\exp\big[\sum_{j=1}^pf_j(a_i,L_0^\ast,l_i)\beta_j\big]}\\
&\qquad=\Bigg(\sum_{i=0}^M\exp\Bigg[\sum_{j=1}^p\big[f_j(a_i,L_0^\ast,l_i)-f_j(a_0,L_0^\ast,l_0)\big]\beta_j\Bigg]\Bigg)^{-1}\\
&\qquad=\Bigg(1+\sum_{i=1}^M\exp\Bigg[\sum_{j=1}^p\big[f_j(a_i,L_0^\ast,l_i)-f_j(a_0,L_0^\ast,l_0)\big]\beta_j\Bigg]\Bigg)^{-1}.
\end{align*}

Thus,
\begin{align*}
&\Exp\Bigg[\frac{\sum_{i=1}^{m}X_{i,j}g(X_i,\tilde{\beta})}{1+\sum_{i=1}^{m}g(X_i,\tilde{\beta})}\Bigg|L^\ast_0=l^\ast,Y_K=1,M=m\Bigg]\\
&\qquad\propto \sum_{l_0,...,l_m}\sum_{a_0,...,a_m}\frac{\sum_{i=1}^{m}[f_k(a_i,l^\ast,l_i)-f_k(a_0,l^\ast,l_0)]\exp\big\{\sum_{k=1}^p[f_k(a_i,l^\ast,l_i)-f_k(a_0,l^\ast,l_0)]\tilde{\beta}_k\big\}}{1+\sum_{i=1}^{m}\exp\big\{\sum_{k=1}^p[f_k(a_i,l^\ast,l_i)-f_k(a_0,l^\ast,l_0)]\tilde{\beta}_k\big\}}\\&\qquad\qquad{}\times \frac{1}{1+\sum_{i=1}^{m}\exp\big\{\sum_{k=1}^p[f_k(a_i,l^\ast,l_i)-f_k(a_0,l^\ast,l_0)]\beta_k\big\}} h(a_0,...,a_M,l_0,...,l_M) \\
&\qquad\propto \sum_{l_0,...,l_m}\sum_{a_0,...,a_m}h(a_0,...,a_M,l_0,...,l_M)\\&\qquad\qquad\times\sum_{i=1}^{m}[f_k(a_i,l^\ast,l_i)-f_k(a_0,l^\ast,l_0)]\frac{\exp\big\{\sum_{k=1}^p[f_k(a_i,l^\ast,l_i)-f_k(a_0,l^\ast,l_0)]\tilde{\beta}_k\big\}}{1+\sum_{i=1}^{m}\exp\big\{\sum_{k=1}^p[f_k(a_i,l^\ast,l_i)-f_k(a_0,l^\ast,l_0)]\tilde{\beta}_k\big\}}\\&\qquad\qquad{}\times \frac{1}{1+\sum_{i=1}^{m}\exp\big\{\sum_{k=1}^p[f_k(a_i,l^\ast,l_i)-f_k(a_0,l^\ast,l_0)]\beta_k\big\}}\\
&\qquad\propto \sum_{\{(l_0,a_0),...,(l_m,a_M)\}}h(a_0,...,a_M,l_0,...,l_M)\\&\qquad\qquad\times\sum_{u=1}^m\sum_{i=1}^{m}[f_k(a_i,l^\ast,l_i)-f_k(a_u,l^\ast,l_u)]\frac{\exp\big\{\sum_{k=1}^p[f_k(a_i,l^\ast,l_i)-f_k(a_u,l^\ast,l_u)]\tilde{\beta}_k\big\}}{1+\sum_{i=1}^{m}\exp\big\{\sum_{k=1}^p[f_k(a_i,l^\ast,l_i)-f_k(a_u,l^\ast,l_u)]\tilde{\beta}_k\big\}}\\&\qquad\qquad{}\times \frac{1}{1+\sum_{i=1}^{m}\exp\big\{\sum_{k=1}^p[f_k(a_i,l^\ast,l_i)-f_k(a_u,l^\ast,l_u)]\beta_k\big\}}\\
&\qquad=\sum_{\{(l_0,a_0),...,(l_m,a_M)\}}h(a_0,...,a_M,l_0,...,l_M)\\&\qquad\qquad\times\sum_{u=1}^m\sum_{i=1}^{m}[f_k(a_i,l^\ast,l_i)-f_k(a_u,l^\ast,l_u)]\frac{\exp\big\{\sum_{k=1}^pf_k(a_i,l^\ast,l_i)\tilde{\beta}_k\big\}}{\sum_{i=0}^{m}\exp\big\{\sum_{k=1}^pf_k(a_i,l^\ast,l_i)\tilde{\beta}_k\big\}}\\&\qquad\qquad{}\times \frac{\exp\big\{\sum_{k=1}^pf_k(a_u,l^\ast,l_u)\beta_k\big\}}{\sum_{i=0}^{m}\exp\big\{\sum_{k=1}^pf_k(a_i,l^\ast,l_i)\beta_k\big\}}\\
&\qquad=\sum_{\{(l_0,a_0),...,(l_m,a_M)\}}h(a_0,...,a_M,l_0,...,l_M)\\&\qquad\qquad\times\Bigg[\sum_{u,i\in\{1,...,m\}:i>u}[f_k(a_i,l^\ast,l_i)-f_k(a_u,l^\ast,l_u)]\frac{\exp\big\{\sum_{k=1}^pf_k(a_i,l^\ast,l_i)\tilde{\beta}_k\big\}}{\sum_{i=0}^{m}\exp\big\{\sum_{k=1}^pf_k(a_i,l^\ast,l_i)\tilde{\beta}_k\big\}}\\&\qquad\qquad{}\times \frac{\exp\big\{\sum_{k=1}^pf_k(a_u,l^\ast,l_u)\beta_k\big\}}{\sum_{i=0}^{m}\exp\big\{\sum_{k=1}^pf_k(a_i,l^\ast,l_i)\beta_k\big\}}\\&\qquad\qquad{}+
\sum_{u,i\in\{1,...,m\}:i<u}[f_k(a_i,l^\ast,l_i)-f_k(a_u,l^\ast,l_u)]\frac{\exp\big\{\sum_{k=1}^pf_k(a_i,l^\ast,l_i)\tilde{\beta}_k\big\}}{\sum_{i=0}^{m}\exp\big\{\sum_{k=1}^pf_k(a_i,l^\ast,l_i)\tilde{\beta}_k\big\}}\\&\qquad\qquad{}\times \frac{\exp\big\{\sum_{k=1}^pf_k(a_u,l^\ast,l_u)\beta_k\big\}}{\sum_{i=0}^{m}\exp\big\{\sum_{k=1}^pf_k(a_i,l^\ast,l_i)\beta_k\big\}}\Bigg]\\
&\qquad=\sum_{\{(l_0,a_0),...,(l_m,a_M)\}}h(a_0,...,a_M,l_0,...,l_M)\\&\qquad\qquad\times\Bigg[\sum_{u,i\in\{1,...,m\}:i>u}[f_k(a_i,l^\ast,l_i)-f_k(a_u,l^\ast,l_u)]\frac{\exp\big\{\sum_{k=1}^pf_k(a_i,l^\ast,l_i)\tilde{\beta}_k\big\}}{\sum_{i=0}^{m}\exp\big\{\sum_{k=1}^pf_k(a_i,l^\ast,l_i)\tilde{\beta}_k\big\}}\\&\qquad\qquad{}\times \frac{\exp\big\{\sum_{k=1}^pf_k(a_u,l^\ast,l_u)\beta_k\big\}}{\sum_{i=0}^{m}\exp\big\{\sum_{k=1}^pf_k(a_i,l^\ast,l_i)\beta_k\big\}}\\&\qquad\qquad{}-
\sum_{u,i\in\{1,...,m\}:i>u}[f_k(a_i,l^\ast,l_i)-f_k(a_u,l^\ast,l_u)]\frac{\exp\big\{\sum_{k=1}^pf_k(a_u,l^\ast,l_u)\tilde{\beta}_k\big\}}{\sum_{i=0}^{m}\exp\big\{\sum_{k=1}^pf_k(a_i,l^\ast,l_i)\tilde{\beta}_k\big\}}\\&\qquad\qquad{}\times \frac{\exp\big\{\sum_{k=1}^pf_k(a_i,l^\ast,l_i)\beta_k\big\}}{\sum_{i=0}^{m}\exp\big\{\sum_{k=1}^pf_k(a_i,l^\ast,l_i)\beta_k\big\}}\Bigg]\\
&\qquad=\sum_{\{(l_0,a_0),...,(l_m,a_M)\}}h(a_0,...,a_M,l_0,...,l_M)\\&\qquad\qquad\times\sum_{u,i\in\{1,...,m\}:i>u}[f_k(a_i,l^\ast,l_i)-f_k(a_u,l^\ast,l_u)]\\&\qquad\qquad{}\times\Bigg[\frac{\exp\big\{\sum_{k=1}^pf_k(a_i,l^\ast,l_i)\tilde{\beta}_k\big\}}{\sum_{i=0}^{m}\exp\big\{\sum_{k=1}^pf_k(a_i,l^\ast,l_i)\tilde{\beta}_k\big\}} \frac{\exp\big\{\sum_{k=1}^pf_k(a_u,l^\ast,l_u)\beta_k\big\}}{\sum_{i=0}^{m}\exp\big\{\sum_{k=1}^pf_k(a_i,l^\ast,l_i)\beta_k\big\}}\\&\qquad\qquad{}-
\frac{\exp\big\{\sum_{k=1}^pf_k(a_u,l^\ast,l_u)\tilde{\beta}_k\big\}}{\sum_{i=0}^{m}\exp\big\{\sum_{k=1}^pf_k(a_i,l^\ast,l_i)\tilde{\beta}_k\big\}} \frac{\exp\big\{\sum_{k=1}^pf_k(a_i,l^\ast,l_i)\beta_k\big\}}{\sum_{i=0}^{m}\exp\big\{\sum_{k=1}^pf_k(a_i,l^\ast,l_i)\beta_k\big\}}\Bigg],
\end{align*}
which is clearly zero when $\tilde{\beta}=\beta$. If follows that
\begin{align*}
\frac{\partial}{\partial \tilde{\beta}_k} f(\tilde{\beta}) &= \Exp\Bigg[ \frac{\sum_{i=1}^{M}X_{i,k}g(X_i,\tilde{\beta})}{1+\sum_{i=1}^{m}g(X_i,\tilde{\beta})}\Bigg|Y_K=1,M>0\Bigg] = 0
\end{align*}
for all $k=1,...,p$ if and only if $\tilde{\beta}=\beta$.
\end{proof}

\end{myAppendix}

\end{document}